\documentclass[a4paper,onecolumn,11pt]{quantumarticle}
\pdfoutput=1

\usepackage[utf8]{inputenc}
\usepackage[english]{babel}
\usepackage[T1]{fontenc}

\usepackage{amsmath,amssymb,amsthm}
\usepackage{graphicx}
\usepackage{physics}
\usepackage{hyperref}
\usepackage[numbers,sort&compress]{natbib}
\usepackage{tikz}
\usetikzlibrary{arrows.meta,positioning,matrix,fit,calc,decorations.pathreplacing}
\usepackage{quantikz}

\usepackage[normalem]{ulem} 
\usepackage{cancel}
\usepackage{multirow}

\theoremstyle{plain}
\newtheorem{theorem}{Theorem}[section]
\newtheorem{lemma}[theorem]{Lemma}

\newtheorem{corollary}[theorem]{Corollary}

\theoremstyle{definition}
\newtheorem{definition}[theorem]{Definition}

\theoremstyle{remark}
\newtheorem{remark}[theorem]{Remark}

\newcommand{\polylog}{\operatorname{polylog}}
\newcommand{\HWC}{\mathrm{HWC}}
\newcommand{\sym}{\mathrm{sym}}
\newcommand{\CNOT}{\mathsf{CNOT}}
\newcommand{\Tof}{\mathsf{Toffoli}}

\newcommand\blfootnote[1]{%
  \begingroup
  \renewcommand\thefootnote{}%
  \footnote{#1}%
  \addtocounter{footnote}{-1}%
  \endgroup
}

\title{Efficient Depth--Ancilla Tradeoffs for Hamming Weight Computation and Symmetric Boolean Functions}

\author{Wei Zi}
\affiliation{Quantum Science Center of Guangdong-Hong Kong-Macao Greater Bay Area, Shenzhen, China}

\author{Pei Yuan}
\affiliation{Independent Researcher, Shenzhen, Guangdong, China}

\author{Junhong Nie}
\affiliation{Shandong University, Jinan, Shandong, China}

\author{Shengyu Zhang}
\affiliation{Independent Researcher, Shenzhen, Guangdong, China}

\begin{document}

\maketitle

\blfootnote{Wei Zi and Pei Yuan contributed equally to this work. 
Emails: Wei Zi (\texttt{ziwei.quantum@outlook.com}); Pei Yuan
(\texttt{
peiyuan0104@gmail.com}); Junhong Nie
(\texttt{123miyoi@gmail.com}); Shengyu Zhang
(\texttt{
shengyuzhang@gmail.com}). Shengyu Zhang is the corresponding
author.}

\begin{abstract}
Hamming weight computation maps an $n$-bit input to the number of ones it contains. It is a basic subroutine in quantum computing, and the core building block for symmetric Boolean functions, whose value depends only on the Hamming weight of the input. Moreover, symmetric Boolean functions are among the most common primitives in quantum computing. Efficient circuits for both problems are therefore important for the efficiency of many quantum algorithms.
We study the depth-ancilla tradeoffs of Hamming weight computation under two qubit connectivity models, all-to-all and two-dimensional nearest-neighbor square grid (2D), in both the standard and dynamic circuit models. In the standard all-to-all model, we obtain depth $O(\log n)$ with a sublinear number of ancillas. In the standard 2D model, we give a circuit of depth $O(\sqrt n)$ with $O(\log^2 n)$ ancillas, and a matching lower
bound showing that $\Theta(\sqrt n)$ is optimal. In both dynamic models, we obtain constant-depth circuits with $O(n^{1+\varepsilon}\polylog n)$ ancillary qubits for every fixed $\varepsilon>0$. All constructions give a smooth depth-ancilla tradeoff, and they also extend to arbitrary symmetric Boolean functions.
\end{abstract}

\section{Introduction}

Quantum computers can solve some important problems that are believed to be intractable for classical computers, such as integer factoring  \cite{365700}. Any quantum algorithm is executed by running a quantum circuit on a quantum device.
Consequently, to improve the performance and efficiency of quantum algorithms, quantum circuit optimization has become an active and widely studied research area \cite{yan2024quantum,kusyk2021survey,karuppasamy2024pqy}. A primary objective is to reduce the circuit depth, which is limited by the coherence time of the quantum device. A related objective is to understand the trade-off between circuit depth and the number of ancilla qubits, since both are limited resources on near-term hardware. In addition,  some implementation schemes of quantum devices such as superconducting qubits \cite{arute2019quantum,kim2023evidence,Wu2021Strong,Google2025willow} only allow 2-qubit gates between certain pairs of qubits. Quantum circuit optimization should therefore also respect the qubit connectivity constraints of the device.

Many quantum algorithms \cite{365700,grover1996fast, gilyen2019quantum, Low2024tradingtgatesdirty, liu2025blockencodinglowgate} rely on the ability to evaluate a Boolean function coherently on a quantum register, either as an oracle that marks target inputs or as an arithmetic subroutine. Because such Boolean subroutines are typically invoked many times within a larger algorithm, the depth and ancilla cost of their quantum circuit implementation directly determines the overall algorithmic efficiency.

Among the many Boolean functions used in quantum algorithms, the Hamming weight function has received considerable attention due to its wide range of applications.  Given an $n$-bit input string, the Hamming weight function returns the number of $1$'s in the string, namely its Hamming weight. The corresponding quantum circuit takes the input on an $n$-qubit input register and records its Hamming weight on a separate $\lceil\log(n+1)\rceil$-qubit output register, while leaving the input register unchanged. Quantum circuits for Hamming weight computation have wide-ranging applications across many directions, including quantum machine learning \cite{li2022quantum,li2023quantum}, quantum state preparation \cite{gosset2026quantum}, symmetric Boolean functions \cite{ZiNieSun2025SymmetricFunctions,TakahashiTani2016Collapse}, quantum simulation \cite{Kivlichan2020Trotterization} and quantum memory \cite{Khan2022EP}.

The optimization of the depth and ancilla count of Hamming weight computation has been extensively studied in the standard quantum circuit model, in which a circuit consists of arbitrary single-qubit and two-qubit gates, with no intermediate measurements and no qubit connectivity constraints. Earlier constructions trade off these two resources in different ways. With only $O(1)$ ancillary qubits, \cite{li2022quantum,li2023quantum} achieve a depth of $O(n\log n)$. Allowing $O(n)$ ancillary qubits brings the depth down to $O(n)$ \cite{chattopadhyay2016low,orts2024quantum}. The depth can be further reduced to $O(\log^2 n)$ using $O(\log n)$ clean and $O(n\log^2 n)$ dirty ancillary qubits \cite{takahashi2021power}. Most recently, Ref.~\cite{ZiNieSun2025SymmetricFunctions} matches this depth of $O(\log^2 n)$ while using no ancillary qubits. If the mid-circuit measurement and classical feedback are allowed in a quantum circuit, it is called the dynamic circuit. Takahashi and Tani~\cite{TakahashiTani2016Collapse} showed that, in the
unbounded-fanout circuit model where a fan-out gate is treated as a primitive,
Hamming weight computation can be implemented in constant depth using
\(O(n\log n)\) ancillary qubits.  To translate this result into the dynamic
circuit model considered here, each unbounded-fanout gate must be realized by
mid-circuit measurements and classical feedforward.  Since an \(n\)-qubit
fan-out can be implemented in \(O(1)\) depth using \(O(n)\) ancillary qubits
\cite{BaumerWoerner2025LongRangeFanout}, the Takahashi--Tani construction gives
a dynamic-circuit implementation of Hamming weight computation with depth
\(O(1)\) and \(O(n^2)\) ancillary qubits.

A Boolean function is said to be 
\emph{symmetric} if its value depends only on the Hamming weight of the input. This seemingly simple structure captures many of the most ubiquitous primitives \cite{TakahashiTani2016Collapse,Buhrman2024statepreparation,HoyerSpalek2005Fanout,Biswal2022Efficient} in quantum computing---including \textsc{or}, \textsc{and}, \textsc{parity}, \textsc{majority}, and \textsc{threshold}. In quantum computing, the efficient circuit construction of 
symmetric Boolean functions is of considerable importance, as many core subroutines and higher-level algorithms repeatedly invoke such functions during execution. In oracle-based 
applications, symmetric Boolean functions constitute the 
standard form of marking functions in some quantum algorithms, such as Grover 
search~\cite{grover1996fast},  
quantum counting \cite{brassard1998quantum} and the quantum algorithm for string matching \cite{niroula2021quantum}. 
The circuit depth optimization of symmetric Boolean functions has been extensively studied in the standard circuit model. For a symmetric Boolean function with $n$-bit input, early constructions in~\cite{perkowski2001regular,perkowski2001regularity} achieve $O(n^2)$ depth using $O(n^2)$ ancillary qubits. Refs.~\cite{maslov2004dynamic,maslov2006efficient} retain the same depth while reducing the ancilla count to $O(n)$, and Ref.~\cite{chattopadhyay2016low} further improves the depth to $O(n \log n)$ with $O(n)$ ancillary qubits. Takahashi and Tani~\cite{takahashi2021power} reduce the depth to $O(\log^2 n)$, but at the cost of $O(\log n)$ clean and $O(n \log^2 n)$ dirty ancillary qubits. At the opposite end of the trade-off, Ref.~\cite{maslov2021quantum} eliminates ancillary qubits entirely, but with a circuit depth of $O(n^2)$. In contrast, the construction in Ref.~\cite{ZiNieSun2025SymmetricFunctions} simultaneously achieves $O(\log^2 n)$ depth and uses only $\lceil \log(n+1) \rceil$ ancillary qubits, optimizing both metrics at once.

\paragraph{Our results.}
This paper studies depth--ancilla tradeoffs for Hamming weight computation in
both all-to-all and two-dimensional nearest-neighbor square-grid (2D) qubit connectivity constraints, in the standard and dynamic circuit models.  The output register of
size $\lceil\log(n+1)\rceil$ is not counted as ancillary workspace.  Tables
\ref{tab:alltoall-results} and \ref{tab:2d-results} summarize the main bounds. A subscript on the asymptotic notation, as in $O_\varepsilon(1)$, indicates that the hidden constant may depend on the subscripted parameter $\varepsilon$. If $\varepsilon$ is a constant, $O_\varepsilon(1)=O(1)$. 

\begin{table}[htpb]
    \centering
    \caption{Circuit depth and ancilla count for Hamming weight computation under the all-to-all qubit connectivity in standard and dynamic circuit model. Here \(n\) is the input size, \(2\le B\le n\), \(r\ge1\) is an integer independent of \(n\), and \(\varepsilon>0\) is fixed.}
    \label{tab:alltoall-results}
    \renewcommand{\arraystretch}{1.15}
    \begin{tabular}{c c c c}
    \hline\hline
    \textbf{reference} & \textbf{model} & \textbf{depth} & \textbf{ancilla count} \\
    \hline
    \cite{li2022quantum,li2023quantum}
        & standard
        & $O(n\log n)$
        & $O(1)$ \\
    \cite{chattopadhyay2016low,orts2024quantum}
        & standard
        & $O(n)$
        & $O(n)$ \\
    \cite{takahashi2021power}
        & standard
        & $O(\log^2 n)$
        & $O(n\log^2 n)$ \\
    \cite{ZiNieSun2025SymmetricFunctions}
        & standard
        & $O(\log^2 n)$
        & $0$ \\
    Thm.~\ref{thm:alltoall-blocked}
        & standard
        & $O(\log^2 B+\log n)$
        & $O((n/B)\log B)$ \\
    Cor.~\ref{cor:alltoall-sublinear}
        & standard
        & $O(\log n)$
        & $o(n)$ \\
    \hline
    \cite{TakahashiTani2016Collapse}
        & dynamic
        & $O(1)$
        & $O(n^2)$ \\
    Thm.~\ref{thm:alltoall-dynamic-one-level}
        & dynamic
        & $O(1)$
        & $O(nB+n^2\log B/B+n\,\mathrm{polylog}\,n)$ \\
    Cor.~\ref{cor:alltoall-dynamic-subquadratic}
        & dynamic
        & $O(1)$
        & $O(n^{3/2}\sqrt{\log n})$ \\
    Thm.~\ref{thm:alltoall-dynamic-recursive}
        & dynamic
        & $O(r)$
        & $O(r\,n^{1+1/r}\,\mathrm{polylog}\,n)$ \\
    Cor.~\ref{cor:alltoall-dynamic-near-linear}
        & dynamic
        & $O_\varepsilon(1)$
        & $O(n^{1+\varepsilon}\,\mathrm{polylog}\,n)$ \\
    \hline\hline
    \end{tabular}
\end{table}

\begin{table}[htbp]
    \centering
    \caption{Circuit depth and ancilla count for Hamming weight computation under the two-dimensional nearest-neighbor square-grid connectivity (2D), in the standard and dynamic circuit models. Here \(n\) is the input size, \(2\le B\le n\), \(r\ge1\) is an integer independent of \(n\), and \(\varepsilon>0\) is fixed.}
    \label{tab:2d-results}
    \renewcommand{\arraystretch}{1.15}
    \begin{tabular}{c c c c}
    \hline\hline
    \textbf{reference} & \textbf{model} & \textbf{depth} & \textbf{ancilla count} \\
    \hline
    Thm.~\ref{thm:hw-2d-polylog-ancilla}
        & standard
        & $O(\sqrt n)$, optimal
        & $O(\log^2 n)$ \\
    \hline
    Thm.~\ref{thm:2d-dynamic-direct}
        & dynamic 
        & $O(1)$
        & $O(n^2)$ \\
    Thm.~\ref{thm:2d-dynamic-blocked}
        & dynamic 
        & $O(1)$
        & $O(nB+n^2\log B/B+n\,\mathrm{polylog}\,n)$ \\
    Cor.~\ref{cor:2d-dynamic-subquadratic}
        & dynamic 
        & $O(1)$
        & $O(n^{3/2}\sqrt{\log n})$ \\
    Thm.~\ref{thm:2d-dynamic-recursive}
        & dynamic 
        & $O(r)$
        & $O(r\,n^{1+1/r}\,\mathrm{polylog}\,n)$ \\
    Cor.~\ref{cor:2d-dynamic-near-linear}
        & dynamic 
        & $O_\varepsilon(1)$
        & $O(n^{1+\varepsilon}\,\mathrm{polylog}\,n)$ \\
    \hline\hline
    \end{tabular}
\end{table}

The all-to-all results give two improvements over the known tradeoffs.  In the
standard model, blocking the ancilla-free Fourier-encoding construction and
then summing the local weights gives logarithmic depth with sublinear ancillary
workspace.  In the dynamic model, the same blocking idea, combined with
constant-depth weighted counting, reduces the ancillary cost of constant-depth
Hamming weight computation from the direct $O(n^2)$ bound to subquadratic and,
by recursion, to $O(n^{1+\varepsilon}\,\mathrm{polylog}\,n)$ for every fixed
$\varepsilon>0$.

The two-dimensional results reveal a sharp separation between
measurement-free locality and dynamic circuits with nonlocal classical
feedforward.
In the measurement-free two-dimensional model, Hamming weight computation has
optimal depth \(\Theta(\sqrt n)\): we give an \(O(\sqrt n)\)-depth construction
using only \(O(\log^2 n)\) clean ancillas, and the matching lower bound follows
from the parity light-cone obstruction.  Thus, without measurements, the
two-dimensional geometry imposes an unavoidable square-root depth barrier.

In contrast, this square-root barrier disappears in the dynamic
two-dimensional model.  Even though the quantum gates are still restricted to
nearest-neighbor interactions on the grid, Hamming weight computation can be
implemented in constant depth.  More strongly, the constant-depth construction
can be achieved with subquadratic ancillary workspace, and the recursive
version uses only $O(n^{1+\varepsilon}\,\mathrm{polylog}\,n)$ ancillas for
every fixed $\varepsilon>0$.  Thus, under nonlocal classical feedforward,
dynamic circuits exhibit a sharp improvement over measurement-free
two-dimensional circuits: the optimal $\Theta(\sqrt n)$ depth barrier is
replaced by constant depth with near-linear ancillary overhead.

The same Hamming-weight primitives give improved circuits for the entire class
of symmetric Boolean functions, rather than only for Hamming weight itself.
For an arbitrary symmetric Boolean function \(f(x)=g(|x|)\), the oracle for
\(f\) is obtained by computing \(|x|\), evaluating \(g\) on the
\(m=\lceil\log(n+1)\rceil\)-qubit weight register, and uncomputing \(|x|\).
Since the nontrivial dependence on the \(n\)-bit input is isolated in the
Hamming-weight computation step, our depth--ancilla tradeoffs for
Hamming weight computation translate directly into tradeoffs for
arbitrary symmetric Boolean oracles.

This yields new bounds in all four settings considered in this paper, as
summarized in Table~\ref{tab:sym}.  In the standard all-to-all model, we obtain
\(O(\log n)\)-depth circuits for arbitrary symmetric Boolean functions using
only \(O(n/\log n)\) clean ancillary qubits, improving the depth over the
previous \(O(\log^2 n)\)-depth constructions while keeping the workspace
sublinear.  In the standard 2D model, we obtain
\(O(\sqrt n)\)-depth circuits using \(O(\log^2 n)\) clean ancillary qubits; this
depth is worst-case optimal because parity is a symmetric Boolean function and
satisfies the usual two-dimensional light-cone lower bound.  Finally, in both
dynamic all-to-all and dynamic 2D models, arbitrary symmetric
Boolean functions can be implemented in \(O_\varepsilon(1)\) depth with
\(O(n^{1+\varepsilon}\mathrm{polylog}\,n)\) clean ancillary qubits.

\begin{table}[htpb]
    \centering
    \caption{Circuit depth and ancilla count for symmetric Boolean functions
    under all-to-all and two-dimensional nearest-neighbor square-grid
    connectivity (2D), in the standard and dynamic circuit models.  Here \(n\) is
    the input size and \(\varepsilon>0\) is fixed.}
    \label{tab:sym}
    \renewcommand{\arraystretch}{1.15}
    \begin{tabular}{c c c c c}
    \hline\hline
    \textbf{reference} & \textbf{model} & \textbf{connectivity}
        & \textbf{depth} & \textbf{ancilla count} \\
    \hline
    \cite{perkowski2001regular,perkowski2001regularity}
        & standard
        & all-to-all
        & \(O(n^2)\)
        & \(O(n^2)\) \\
    \cite{maslov2004dynamic,maslov2006efficient}
        & standard
        & all-to-all
        & \(O(n^2)\)
        & \(O(n)\) \\
    \cite{chattopadhyay2016low}
        & standard
        & all-to-all
        & \(O(n\log n)\)
        & \(O(n)\) \\
    \cite{takahashi2021power}
        & standard
        & all-to-all
        & \(O(\log^2 n)\)
        & \(O(n\log^2 n)\) \\
    \cite{maslov2021quantum}
        & standard
        & all-to-all
        & \(O(n^2)\)
        & \(0\) \\
    \cite{ZiNieSun2025SymmetricFunctions}
        & standard
        & all-to-all
        & \(O(\log^2 n)\)
        & \(\lceil\log(n+1)\rceil\) \\
    \hline
    \multirow{4}{*}{Thm.~\ref{thm:symmetric-functions}}
        & standard
        & all-to-all
        & \(O(\log n)\)
        & \(O(n/\log n)\) \\

        & dynamic
        & all-to-all
        & \(O_\varepsilon(1)\)
        & \(O(n^{1+\varepsilon}\,\mathrm{polylog}\,n)\) \\

        & standard
        & 2D
        & \(O(\sqrt n)\)
        & \(O(\log^2 n)\) \\

        & dynamic
        & 2D
        & \(O_\varepsilon(1)\)
        & \(O(n^{1+\varepsilon}\,\mathrm{polylog}\,n)\) \\
    \hline\hline
    \end{tabular}
\end{table}

\paragraph{Technical overview.}
The technical core of the paper is Hamming weight computation; the
symmetric-function bounds are obtained by computing \(|x|\), evaluating
\(g(|x|)\) on the \(m=\lceil\log(n+1)\rceil\)-qubit weight register, and
uncomputing \(|x|\).
We therefore focus in this overview on the Hamming-weight constructions,
which are all based on the same high-level reduction shown in
Fig.~\ref{fig:technical-overview}.

At a high level, the constructions first reduce Hamming weight computation
to weighted counting.  We divide the input into blocks, compute local Hamming
weights \(w_\ell\), and write each local weight as
\(w_\ell=\sum_j2^j b_{\ell,j}\).  Then the global Hamming weight is
\[
    |x|=\sum_{\ell,j}2^j b_{\ell,j},
\]
so the remaining task is to compute a weighted sum of the bits \(b_{\ell,j}\).
The four settings considered below differ in how the local weights and this
weighted-counting step are implemented.

\begin{figure}[htbp]
\centering
\begin{tikzpicture}[
    x=1cm,y=1cm,
    every node/.style={font=\small},
    box/.style={draw,thick,rounded corners=2pt,fill=gray!8,
        minimum width=2.25cm,minimum height=0.82cm,align=center},
    bluebox/.style={draw,thick,rounded corners=2pt,fill=blue!5,
        minimum width=2.35cm,minimum height=0.82cm,align=center},
    greenbox/.style={draw,thick,rounded corners=2pt,fill=green!7,
        minimum width=2.20cm,minimum height=0.82cm,align=center},
    smallblock/.style={draw,thick,rounded corners=1.5pt,fill=gray!10},
    arrow/.style={->,thick},
    dashedarrow/.style={->,thick,dashed,blue!70},
    >=Latex
]

\node[anchor=west] at (0.0,4.35) {(a) Blocking};
\node[anchor=west] at (4.20,4.35) {(b) Dynamic counting};
\node[anchor=west] at (9.10,4.35) {(c) Two-dimensional layout};

\node[box] (a1) at (1.25,3.35) {input bits};
\node[box] (a2) at (1.25,1.95) {local weights\\$w_1,\ldots,w_L$};
\node[greenbox] (a3) at (1.25,0.55) {weighted sum\\$\sum_{\ell,j}2^j b_{\ell,j}$};

\draw[arrow] (a1.south) -- (a2.north);
\draw[arrow] (a2.south) -- (a3.north);

\node[box] (b1) at (5.60,3.35) {weighted bits\\$z_i,a_i$};
\node[bluebox] (b2) at (5.60,1.95) {phase states\\$\ket{\phi_k(S)}$};
\node[greenbox] (b3) at (5.60,0.55) {decode\\$\ket{S}$};

\draw[arrow] (b1.south) -- (b2.north);
\draw[arrow] (b2.south) -- (b3.north);

\draw[thick] (9.10,0.15) rectangle (14.40,3.65);

\node at (11.75,3.35) {local patches};

\draw[smallblock] (9.45,2.55) rectangle ++(0.62,0.42);
\draw[smallblock] (10.75,2.55) rectangle ++(0.62,0.42);
\draw[smallblock] (12.05,2.55) rectangle ++(0.62,0.42);
\draw[smallblock] (10.10,1.65) rectangle ++(0.62,0.42);
\draw[smallblock] (11.40,1.65) rectangle ++(0.62,0.42);

\draw[thick,rounded corners=2pt,fill=blue!5]
    (10.20,0.55) rectangle (13.00,1.20);
\node at (11.60,0.875) {global counting};

\draw[thick,rounded corners=2pt,fill=green!7]
    (13.30,0.55) rectangle (14.10,1.20);
\node at (13.70,0.875) {$\ket{|x|}$};
\draw[arrow] (13.00,0.875) -- (13.30,0.875);

\draw[dashedarrow] (9.76,2.55) .. controls (9.60,2.05) and (10.05,1.55) .. (10.50,1.22);
\draw[dashedarrow] (10.41,1.65) .. controls (10.48,1.45) and (10.72,1.30) .. (10.95,1.22);
\draw[dashedarrow] (11.06,2.55) .. controls (11.00,2.05) and (11.15,1.55) .. (11.40,1.22);
\draw[dashedarrow] (11.71,1.65) .. controls (11.78,1.45) and (11.85,1.30) .. (11.90,1.22);
\draw[dashedarrow] (12.36,2.55) .. controls (12.45,2.05) and (12.62,1.55) .. (12.45,1.22);

\end{tikzpicture}
\caption{Overview of the constructions.  Panel (a) shows the blocked reduction
from Hamming weight computation to weighted counting.  Panel (b) shows the
dynamic weighted-counting primitive: fan-out prepares redundant Fourier phase
states for $S=\sum_i a_i z_i$, and a constant-depth decoder writes $\ket{S}$.
Panel (c) shows the two-dimensional dynamic layout, where local outputs are
copied by measurement-based long-range CNOTs to a global weighted-counting
patch.}
\label{fig:technical-overview}
\end{figure}
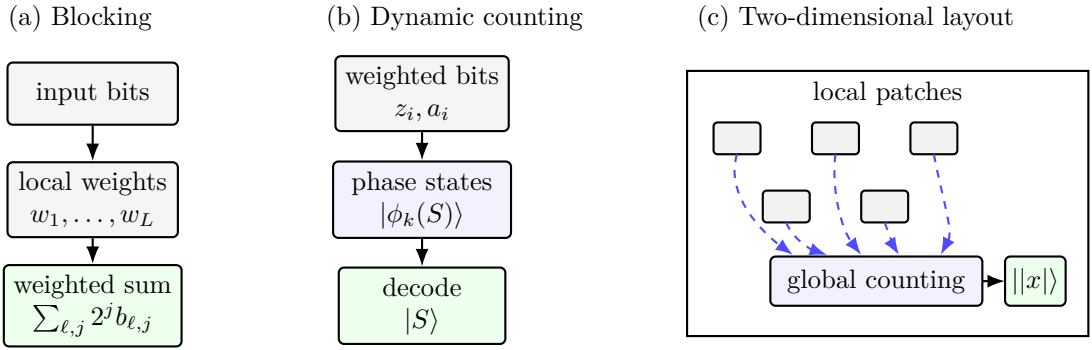

In the standard all-to-all model, the local Hamming weights are computed by the
ancilla-free Fourier-encoding construction, and the remaining weighted sum is
performed by a carry-save population-count circuit.  This gives the blocked
tradeoff in Theorem~\ref{thm:alltoall-blocked}; choosing the block size
appropriately yields logarithmic depth with sublinear ancillary workspace.

In the dynamic model, the central primitive is weighted counting.  Given bits
$z_i$ with fixed weights $a_i$ and $S=\sum_i a_i z_i$, we prepare the redundant
Fourier phase states needed by the Takahashi--Tani decoder
\cite{TakahashiTani2016Collapse}.  For each phase state,
measurement-based fan-out creates a cat state
\cite{BaumerWoerner2025LongRangeFanout}, and weighted controlled phase gates
accumulate the phase corresponding to $S$.  The decoder then writes the binary
representation of $S$ in constant depth.  This primitive costs
$O(MR+R\,\mathrm{polylog}\,R)$ ancillary qubits when there are $M$ weighted
bits and the sum lies in $0,\ldots,R$.

The dynamic tradeoffs follow by applying this primitive recursively.  A
one-level blocking gives a constant-depth construction with
$O(n^{3/2}\sqrt{\log n})$ ancillary qubits.  More generally, an $r$-level
pyramid computes weights of larger and larger blocks; each level is a
constant-depth weighted-counting layer, and each level uses
$O(n^{1+1/r}\,\mathrm{polylog}\,n)$ ancillary qubits.  This gives depth $O(r)$
and ancillary cost $O(rn^{1+1/r}\,\mathrm{polylog}\,n)$.

Without measurements, two-dimensional circuits cannot use long-range
fan-out.  We therefore localize the Fourier-encoding construction of Hamming
weight computation~\cite{ZiNieSun2025SymmetricFunctions}.  The key primitive is
a multi-output phase gadget: within each \(B\)-bit block, the \(O(\log B)\)
output qubits jointly accumulate the phases
\(\exp(i\theta_j y_j\sum_i x_i)\) by geometry-aware routing on a two-dimensional
patch.  Choosing \(B=\Theta(n/\log n)\), we route the remaining
\(O(\log^2 n)\) local-weight qubits to a small output region and sum them by a
nearest-neighbor carry-save circuit.  This yields depth \(O(\sqrt n)\) with
\(O(\log^2 n)\) clean ancillas, and this depth is optimal by the parity
light-cone lower bound~\cite{MooreNilsson2001ParallelQuantum,Rosenbaum2013NearestNeighbor}.

Finally, in the dynamic two-dimensional model, we implement the
weighted-counting primitive as an explicit two-dimensional patch.  The patch
arranges the weighted input bits in one direction and the Fourier-decoding
branches in the other: row fan-out copies each weighted input bit to all
branches, column fan-out prepares one cat state per branch, and local diagonal
gates accumulate the weighted phases.  The Takahashi--Tani decoder then outputs
the sum in constant depth.  We then plug this patch into the same blocked and
recursive framework as in the all-to-all dynamic construction.  At each level,
local or child weight bits are copied to reserved input sites on the boundary
of weighted-counting patches through vertex-disjoint routing corridors, the
weighted sums are computed in constant depth, and the copies are uncomputed.
Thus each recursive level has constant depth, while the patch area and corridor
area are absorbed into the same ancillary bounds as in the all-to-all dynamic
case.  This gives constant-depth subquadratic-ancilla constructions and, after
\(r\) recursive levels, depth \(O(r)\) with
\(O(rn^{1+1/r}\mathrm{polylog}\,n)\) clean ancillas; for every fixed
\(\varepsilon>0\), this yields \(O_\varepsilon(1)\) depth with
\(O(n^{1+\varepsilon}\mathrm{polylog}\,n)\) clean ancillas.

The rest of this paper is organized as follows. In Section \ref{sec:preliminaries}, we introduce notation and definitions of problems, and review some previous results. In Sections \ref{sec:all-to-all} and \ref{sec:two-dimensional}, we give the circuit constructions for Hamming weight computation in both standard and dynamic circuit models, with all-to-all and two-dimensional nearest-neighbor square-grid qubit connectivity constraints. We apply the circuit constructions of the Hamming weight computation to symmetric Boolean functions in Section \ref{sec:symmetric-functions}. We conclude and discuss in Section \ref{sec:discussion}.

\section{Preliminaries and Circuit Models}
\label{sec:preliminaries}

\paragraph{Notation} We use the following notation and conventions. For an $n$-bit string  $x=(x_1,\ldots,x_n)\in\{0,1\}^n$, we denote its Hamming weight by $|x|=\sum_{i=1}^n x_i$. For any $b\in\{0,1\}$ and $r\in\mathbb{N}$, $b^r$ denotes an $r$-bit string $b\cdots b\in\{0,1\}^r$. A CNOT gate $\CNOT^p_q$ acting on qubits $p,q$ satisfying $\CNOT^p_q\ket{x}_p\ket{y}_q=\ket{x}_p\ket{x\oplus y}_q$ for any $x,y\in\{0,1\}$. A Toffoli gate $\Tof(a,b;c)$ satisfying $\Tof(a,b;c)\ket{a}\ket{b}\ket{c}=\ket{a}\ket{b}\ket{ab\oplus c}$ for any $a,b,c\in\{0,1\}$.

We use asymptotic notation with respect to \(n\).  All logarithms are base two. A subscript on the asymptotic notation, as in $O_\varepsilon(1)$, indicates that the hidden constant may depend on the subscripted parameter $\varepsilon$.

\subsection{Hamming Weight Computation and Resource Measures}
We introduce the definition of Hamming weight computation (HWC) as follows.
\begin{definition}[Hamming weight computation, $\HWC_n$]
   The \(n\)-bit Hamming weight
computation, denoted by \(\HWC_n\), is a unitary transformation satisfying
\begin{equation}\label{eq:hwc}
    U_{\HWC}^{(n)}:
    \ket{x}\ket{0^m}
    \longmapsto
    \ket{x}\ket{|x|},
    \qquad \forall x\in\{0,1\}^n, \quad
    m=\lceil\log(n+1)\rceil .
\end{equation}
Here the second register is the output register.
\end{definition}
For any $|x| \in \{0, 1, \ldots, n\}$, we write 
$\ket{|x|} = \ket{k_{m-1}, k_{m-2},\cdots, k_0}$, 
where
$|x| = \sum_{j=0}^{m-1} k_j \cdot 2^j$ is the binary representation of $|x|$.

Throughout the paper, we distinguish the output register from ancillary
workspace.  The \(m\)-qubit output register \(\ket{0^m}\) is part of the
specified input and is not counted as ancillary workspace.  By the ancillary
qubit count, we mean only additional clean workspace qubits used during the
computation.

The depth of a circuit is the number of parallel layers of gates, where gates
with disjoint supports may be applied simultaneously. The goal of this paper is to construct quantum circuits that implement target tasks exactly, while minimizing the circuit depth and using as few ancillary qubits as possible.

\subsection{Circuit Models and Qubit Connectivity}
We adopt two circuit models and two types of qubit connectivity in this section.

\paragraph{Circuit model} We consider two quantum circuit models, the standard and the dynamic circuit models.
\begin{itemize}
    \item A \textit{standard circuit} consists of arbitrary one- and two-qubit gates. Mid-circuit measurements and classical feedforward are not allowed in this model.

    \item  A \textit{dynamic circuit} consists of arbitrary one- and two-qubit gates, together with mid-circuit measurements and classical feedforward. In this model, qubits may be measured in the computational basis during the computation, and subsequent gates may be chosen according to the classical measurement outcomes.
\end{itemize}

Ancillary qubits are permitted in the circuit models defined above. A \textit{clean ancillary qubit} is initialized to the state $\ket{0}$. A \textit{dirty ancillary qubit} (also called a borrowed ancillary qubit) is in an arbitrary, unknown initial state, possibly entangled with other parts of the system. Both the clean and dirty ancillary qubits need to be restored to their original states after the quantum circuits. 

\paragraph{Qubit connectivity} We further consider two types of qubit connectivity, all-to-all connectivity and two-dimensional square-grid connectivity.
\begin{itemize}
    \item The \textit{all-to-all connectivity} allows a two-qubit gate to be applied to any pair of qubits.

    \item The \textit{two-dimensional square-grid connectivity} (2D) arranges the qubits on the vertices of a two-dimensional square-grid and restricts every two-qubit gate to a pair of qubits that are adjacent on the grid.
\end{itemize}
The square-grid nearest-neighbor model is a standard abstraction for studying
geometric locality in quantum circuits~\cite{MooreNilsson2001ParallelQuantum,Rosenbaum2013NearestNeighbor}. Our upper bounds are stated for the square-grid.  They also apply, up to constant-factor changes in depth and area, to architecture families that can simulate the square-grid with constant overhead.

In the rest of this paper, we study the circuits under four settings: the \textbf{standard all-to-all model}, the \textbf{standard 2D model}, the \textbf{dynamic all-to-all model}, and the \textbf{dynamic 2D model}. We present circuit constructions and resource analyses for each of these four settings.

\paragraph{Fan-out gate} We use this additional power only through measurement-based implementations of fan-out and the associated cat-state recovery procedure.

\begin{definition}[Fan-out]\label{def:fanout}
    The fan-out $F_r$ on $r+1$ qubits is the unitary transformation satisfying
    \begin{equation}
        F_r:\ket{b}\ket{0^r}
    \longmapsto
    \ket{b}\ket{b^r},
    \qquad \forall b\in\{0,1\}.
    \end{equation}
\end{definition}

\begin{lemma}[Fan-out, \cite{BaumerWoerner2025LongRangeFanout}]\label{lem:fanout-path}
    The fan-out $F_r$ can be implemented with constant depth and $O(r)$ ancillary qubits in the dynamic circuit model on a one-dimensional nearest-neighbor architecture.
\end{lemma}

We also use the following \(X\)-basis measurement recovery procedure for cat states: after the fan-out value has been used, the fan-out leaves are measured and released, leaving the root qubit restored up to a Pauli correction determined by the measurement outcomes~\cite{ZiNieSun2025MFStatePreparation}.

After introducing these primitives, we describe some constructions in the unbounded fan-out language, following the standard fan-out circuit model used in prior work~\cite{HoyerSpalek2005Fanout,TakahashiTani2016Collapse}.  This is
only a shorthand for the corresponding dynamic implementation: the fan-out targets, recovery measurements, and ancillary qubits required to realize fan-out are still included in the ancillary-qubit count. 

\subsection{Basic Hamming Weight Primitives}

We record three known Hamming-weight primitives that will be used as building blocks.

\begin{lemma}[Ancilla-free Hamming weight computation, \cite{ZiNieSun2025SymmetricFunctions}]
\label{lem:prelim-hw-tcad}
In the
standard all-to-all model, \(\HWC_n\) can be implemented  with depth \(O(\log^2 n)\) and with no additional
clean ancillary qubits beyond the output register
.
\end{lemma}

The construction is based on Fourier phase encoding: the value \(|x|\) is
encoded into phases of the output register, and an inverse quantum Fourier
transform then converts this phase encoding into the binary representation
\(\ket{|x|}\).

\begin{lemma}[Carry-save Hamming weight computation, \cite{Gossett1998CarrySave}]
\label{lem:prelim-carry-save}
In the
standard all-to-all model, \(\HWC_n\) can be implemented  with depth \(O(\log n)\) and with \(O(n)\) additional
clean ancillary qubits.
\end{lemma}

This follows from the standard reversible carry-save, or Wallace-tree,
population-count construction.  Carry-save techniques have long been used in quantum arithmetic to parallelize additions and defer carry propagation \cite{Gossett1998CarrySave}.  In the present setting, parallel \(3\)-to-\(2\) counters reduce the number of active weighted bits by a constant factor per layer, and the remaining weighted bits are then added into the output register.  Related Hamming-weight primitives appear, for example, in Hamming-weight phasing for fault-tolerant simulation \cite{Kivlichan2020Trotterization} and in logarithmic-depth circuits for
Hamming-weight projections
\cite{RethinasamyLaBordeWilde2024HammingWeight}.

\begin{lemma}[Takahashi--Tani counting, \cite{TakahashiTani2016Collapse}]
\label{lem:prelim-tt-counting}
In the dynamic all-to-all model, \(\HWC_n\) can be implemented with constant
depth and with \(O(n^2)\) additional clean ancillary qubits.
\end{lemma}

Takahashi and Tani gave an exact constant-depth Hamming weight construction in the
unbounded fan-out model~\cite{TakahashiTani2016Collapse}.  Since fan-out can be
implemented in constant depth using mid-circuit measurements and classical
feedforward with linear ancillary overhead
\cite{BaumerWoerner2025LongRangeFanout}, their construction gives the stated
dynamic-circuit primitive.  Conceptually, it prepares a redundant Fourier
encoding of the Hamming weight and decodes all output bits in parallel by
trying all lower-bit guesses.

\section{All-to-All Constructions for Hamming Weight Computation}
\label{sec:all-to-all}

In this section, we prove the all-to-all constructions.  We first combine the
ancilla-free Hamming-weight construction from Lemma~\ref{lem:prelim-hw-tcad}
with carry-save summation to obtain logarithmic depth with sublinear clean
workspace.  We then use measurement-based fan-out to obtain constant-depth
constructions with nearly linear ancillary cost.

\subsection{Blocked Hamming Weight Computation in Standard Circuit Model}

\begin{theorem}[Blocked HWC]
\label{thm:alltoall-blocked}
Let \(2\le B\le n\).  In the standard all-to-all model, \(\HWC_n\) can be implemented with depth
\(O(\log^2 B+\log n)\) and with \(O((n/B)\log B)\) additional clean ancillary
qubits.
\end{theorem}

\begin{proof}
Partition the \(n\) input bits $x\in\{0,1\}^n$ into \(L=\lceil n/B\rceil\) blocks, each of
size at most \(B\).  Let \(h_\ell\) be the Hamming weight of the \(\ell\)-th
block.  By Lemma~\ref{lem:prelim-hw-tcad}, each local Hamming weight can be
computed in depth \(O(\log^2 B)\) using no additional clean workspace beyond
its local output register.  Since the blocks are disjoint, all local weights
can be computed in parallel:
\[
    \ket{x}\ket{0}^{\otimes L\lceil\log(B+1)\rceil}
    \mapsto
    \ket{x}\ket{h_1}\cdots\ket{h_L}, \quad \forall x\in\{0,1\}^n.
\]
The total number of local-weight qubits is $M \le  L \cdot \lceil\log(B+1)\rceil=O\!\left(\frac nB\log B\right)$.

It remains to compute \(h_1+\cdots+h_L=|x|\) into the final output register.
Write each local weight as \(h_\ell=\sum_j2^j b_{\ell,j}\).  Then $|x|=\sum_{\ell,j}2^j b_{\ell,j}$.
Thus, the bits \(b_{\ell,j}\) can be viewed as \(M\) weighted bits. The
standard carry-save population-count construction applies to these weighted
bits: parallel \(3\)-to-\(2\) counters reduce the number of active weighted
bits by a constant factor per layer, and the remaining weighted bits are added
into the \(m\)-qubit output register.  By Lemma~\ref{lem:prelim-carry-save},
this step has depth \(O(\log M)\le O(\log n)\) and uses \(O(M)\) clean
workspace.

Finally, reverse the carry-save summation and then reverse the local
Hamming-weight computations.  This restores all intermediate registers to
\(\ket0\), while leaving the final output register in the state
\(\ket{|x|}\).  Hence the total depth is \(O(\log^2 B+\log n)\), and the
total additional clean workspace is $ O(M)=O\!\left(\frac nB\log B\right)$.
\end{proof}

\begin{corollary}[Sublinear-ancilla logarithmic-depth HWC]
\label{cor:alltoall-sublinear}
In the standard all-to-all model,
\(\HWC_n\) can be implemented with depth \(O(\log n)\) and with $o(n)$
additional clean ancillary qubits.
\end{corollary}

\begin{proof}
Apply Theorem~\ref{thm:alltoall-blocked} with
\(B=2^{\sqrt{\log n}}\).  Then \(\log^2 B=\log n\), and $\frac nB\log B
    =
    \frac{n\sqrt{\log n}}{2^{\sqrt{\log n}}}=o(n)$.
The claimed bounds follow.
\end{proof}

\subsection{Weighted Counting in Dynamic Circuit Model }

The next lemma is the main primitive for the dynamic constructions.  It uses
the parallel decoding technique from the exact fan-out counting construction
of Takahashi and Tani~\cite[Sec.~4.1 and App.~A.3]{TakahashiTani2016Collapse}.  Their
construction decodes a redundant Fourier encoding of an ordinary Hamming
weight in constant depth.  Here we use the same decoding idea, but prepare the
redundant Fourier encoding for a weighted sum.

\begin{lemma}[Weighted counting]
\label{lem:weighted-counting}
Let \(z_1,\ldots,z_M\in\{0,1\}\) be input bits with fixed nonnegative integer
weights \(a_1,\ldots,a_M\).  Suppose weighted sum  $S=\sum_{i=1}^M a_i z_i$
satisfies \(0\le S\le R\).  Then the weighted counting map
\[
    \ket{z}\ket{0^q}\mapsto\ket{z}\ket{S},
    \qquad z=(z_1,z_2,\ldots,z_M)\in\{0,1\}^M, \quad q=\lceil\log(R+1)\rceil,
\]
can be implemented with constant depth in the dynamic all-to-all model, using $ O\bigl(MR+R\,\mathrm{polylog}\,R\bigr)$ additional clean ancillary qubits.
\end{lemma}

\begin{proof}
The case \(R=0\) is trivial, so assume \(R\ge1\).  Let
\(q=\lceil\log(R+1)\rceil\). For \(t\in\{0,\ldots,R\}\) and
\(0\le k<q\), define
\[
    \ket{\phi_{k,y}(t)}
    =
    \frac{\ket0+e^{2\pi i t/2^{k+1}}\ket1}{\sqrt2}, \quad \forall (k,y)\in \mathcal{I},
\]
where $\mathcal{I}=\{(k,y): 0\le k <q, y\in\{0,1\}^k\}$ is the index set of phase-state branches.
The number of branches is
\[
    N_{\mathrm{br}}=\sum_{k=0}^{q-1}2^k=2^q-1=O(R).
\]
Given the corresponding family of states \(\ket{\phi_{k,y}(t)}\), the
Takahashi--Tani decoder outputs the binary representation of \(t\) in constant
fan-out depth using \(O(R\,\mathrm{polylog}\,R)\) ancillary qubits. 

It remains to prepare the same redundant phase encoding for the weighted sum
\(S=\sum_i a_i z_i\).

Fix one branch \((k,y)\), and let $\omega_k=e^{2\pi i/2^{k+1}}$.
Since 
$ 
    e^{2\pi iS/2^{k+1}}
    =
    \omega_k^S
    =
    \prod_{i=1}^M \omega_k^{a_i z_i},
$ 
it suffices to make the \(i\)-th input bit contribute the phase
\(\omega_k^{a_i z_i}\) to the \(\ket1\)-component of a phase qubit.

For this branch, start with one phase qubit in the state \(\ket+\).  Using fan-out $F_M$ implemented with constant depth and $O(M)$ ancillary qubits (Lemma \ref{lem:fanout-path}), create a cat state with \(M\) leaves,
\[
    \ket+ \ket{0^M}
    \longmapsto
    \frac{\ket0^{\otimes(M+1)}+\ket1^{\otimes(M+1)}}{\sqrt2}.
\]
For each \(i\), apply the two-qubit diagonal gate
\[
    D_i^{(k)}:
    \ket{z_i}\ket{\ell_i}
    \longmapsto
    \omega_k^{a_i z_i\ell_i}\ket{z_i}\ket{\ell_i}, \quad \forall z_i,\ell_i\in\{0,1\}
\]
between the input bit \(\ket {z_i}\) and the \(i\)-th cat leaf \(\ket {\ell_i}\).  Equivalently,
\(D_i^{(k)}\) is a controlled phase gate whose nontrivial phase is
\(\omega_k^{a_i}\).  A schematic example for \(M=3\) is shown in
Fig.~\ref{fig:weighted-phase-branch}.

\begin{figure}[t]
\centering
\begin{quantikz}[row sep=0.28cm, column sep=0.42cm]
\lstick{$\ket{+}$}
    & \gate[wires=4]{F_3}
    & \qw
    & \qw
    & \qw
    & \gate[wires=4]{F_3^\dagger}
    & \qw
    & \rstick{$\ket{\phi_k(S)}$} \\
\lstick{$\ket0$}
    &
    & \gate{P_1^{(k)}}
    & \qw
    & \qw
    &
    & \qw
    & \rstick{$\ket0$} \\
\lstick{$\ket0$}
    &
    & \qw
    & \gate{P_2^{(k)}}
    & \qw
    &
    & \qw
    & \rstick{$\ket0$} \\
\lstick{$\ket0$}
    &
    & \qw
    & \qw
    & \gate{P_3^{(k)}}
    &
    & \qw
    & \rstick{$\ket0$} \\
\lstick{$\ket{z_1}$}
    & \qw
    & \ctrl{-3}
    & \qw
    & \qw
    & \qw
    & \qw
    & \rstick{$\ket{z_1}$} \\
\lstick{$\ket{z_2}$}
    & \qw
    & \qw
    & \ctrl{-3}
    & \qw
    & \qw
    & \qw
    & \rstick{$\ket{z_2}$} \\
\lstick{$\ket{z_3}$}
    & \qw
    & \qw
    & \qw
    & \ctrl{-3}
    & \qw
    & \qw
    & \rstick{$\ket{z_3}$}
\end{quantikz}
\caption{Weighted phase preparation for one branch \((k,y)\) when \(M=3\).
Here \(P_i^{(k)}=\operatorname{diag}(1,\omega_k^{a_i})\), so the controlled
operation between \(z_i\) and the \(i\)-th cat leaf implements
\(D_i^{(k)}:\ket{z_i}\ket{\ell_i}\mapsto
\omega_k^{a_i z_i\ell_i}\ket{z_i}\ket{\ell_i}\).}
\label{fig:weighted-phase-branch}
\end{figure}
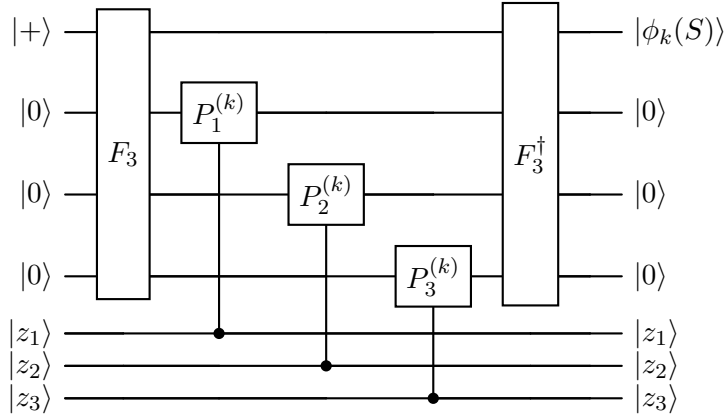

The gates \(D_i^{(k)}\) act on distinct leaf qubits and, after copying the
input bits to this branch if necessary, can be performed in one parallel
layer.  On the all-zero component of the cat state, every \(\ell_i=0\), so no
phase is accumulated.  On the all-one component, every \(\ell_i=1\), and the
accumulated phase is
$ 
    \prod_{i=1}^M\omega_k^{a_i z_i}
    =
    \omega_k^{\sum_i a_i z_i}
    =
    e^{2\pi iS/2^{k+1}} .
$ 
Thus the cat state becomes
\[
    \frac{
        \ket0^{\otimes(M+1)}
        +
        e^{2\pi iS/2^{k+1}}\ket1^{\otimes(M+1)}
    }{\sqrt2}.
\]
Applying the inverse fan-out restores the leaves to \(\ket0\) and leaves the
original phase qubit in the state \(\ket{\phi_{k,y}(S)}\).  Hence one branch can
be prepared in constant depth using \(O(M)\) ancillary qubits.

We prepare all \(N_{\mathrm{br}}=O(R)\) branches in parallel.  If needed, each
input bit \(z_i\) is first fanned out to all branches in which it participates.
The input copies and cat-state leaves together use
\(O(MN_{\mathrm{br}})=O(MR)\) ancillary qubits, and the entire weighted phase preparation has constant depth.

Finally, apply the Takahashi--Tani decoder to the prepared redundant encoding.
If temporary registers are produced, we use the standard
compute--copy--uncompute procedure to write \(\ket S\) into the final output
register and restore all temporary registers to \(\ket0\).  The total
ancillary cost is therefore $O(MR)+O(R\,\mathrm{polylog}\,R) = O\bigl(MR+R\,\mathrm{polylog}\,R\bigr)$,
and the depth remains constant.
\end{proof}

\subsection{Blocked and Recursive Hamming Weight Computation in Dynamic Circuit Model}

The following theorem shows a one-level blocked Hamming weight computation in a dynamic circuit.
\begin{theorem}
\label{thm:alltoall-dynamic-one-level}
Let \(2\le B\le n\).  In the dynamic all-to-all model, \(\HWC_n\) can be
implemented with constant depth and with $O\!\left(nB+\frac{n^2\log B}{B}+n\,\mathrm{polylog}\,n\right)$
additional clean ancillary qubits.
\end{theorem}

\begin{proof}
Partition the input into \(L=\lceil n/B\rceil\) blocks of size at most \(B\).
Let \(w_\ell\) be the Hamming weight of the \(\ell\)-th block.  By
Lemma~\ref{lem:prelim-tt-counting}, each weight \(w_\ell\) can be
computed in constant depth using \(O(B^2)\) clean ancillary qubits.  Since the
blocks are disjoint, all local computations are performed in parallel and use total workspace $L\cdot O(B^2)=O(nB)$.
The local-weight registers themselves occupy
$ M=L\lceil\log(B+1)\rceil=O\!\left(\frac nB\log B\right)$
qubits.

Now write \(w_\ell=\sum_j2^j b_{\ell,j}\), then $|x|=\sum_\ell w_\ell=\sum_{\ell,j}2^j b_{\ell,j}$.
Thus the bits \(b_{\ell,j}\) form \(M=O((n/B)\log B)\) weighted input bits,
and their weighted sum lies in the range \(0\le |x|\le n\).  Applying
Lemma~\ref{lem:weighted-counting} with \(R=n\) computes this weighted sum in
constant depth using
$ 
    O\bigl(Mn+n\,\mathrm{polylog}\,n\bigr)
    =
    O\!\left(\frac{n^2\log B}{B}+n\,\mathrm{polylog}\,n\right)
$
additional clean ancillary qubits.

After the global Hamming weight has been written into the final output
register, we reverse the local Hamming weight computations to restore all
local-weight registers and local workspaces to \(\ket0\).  The total depth is
constant, and the total clean ancillary cost is
$ 
    O(nB)+O\!\left(\frac{n^2\log B}{B}+n\,\mathrm{polylog}\,n\right).
$ 
\end{proof}

\begin{corollary}[Subquadratic-ancilla constant-depth HWC]
\label{cor:alltoall-dynamic-subquadratic}
In the dynamic all-to-all model, $\HWC_n$ can be implemented with constant depth and
\(O(n^{3/2}\sqrt{\log n})\) additional clean ancillary qubits.
\end{corollary}

\begin{proof}
Apply Theorem~\ref{thm:alltoall-dynamic-one-level} with
\(B=\Theta(\sqrt{n\log n})\),  then
$
    nB=O(n^{3/2}\sqrt{\log n})
$ and 
$
    \frac{n^2\log B}{B}=O(n^{3/2}\sqrt{\log n}).
$ 
The remaining \(n\,\mathrm{polylog}\,n\) term is absorbed by this bound.
\end{proof}

Before the recursive construction, we fix the polylogarithmic loss in
Lemma~\ref{lem:weighted-counting}.  Namely, let \(p\ge1\) be a constant such
that the weighted-counting primitive uses
\(O(MR+R(\log R)^p)\) additional clean ancillary qubits; such a \(p\) exists
by Lemma~\ref{lem:weighted-counting}.

The following theorem shows a depth-ancilla tradeoff for HWC.
\begin{theorem}
\label{thm:alltoall-dynamic-recursive}
For every integer \(r\ge1\), $\HWC_n$ can
be implemented in the dynamic all-to-all model with depth $O(r)$
and with
$O\!\left(r\,n^{1+1/r}\,\mathrm{polylog}\,n\right)$ 
additional clean ancillary qubits.
\end{theorem}

\begin{proof}
We use an \(r\)-level pyramid construction.  For simplicity of notation, ignore
integer roundings in the block sizes; replacing them by ceilings changes only
constant factors.  Let $B_j=n^{j/r}$ for $j=0,1,\ldots,r$,
so \(B_0=1\) and \(B_r=n\).  At level \(j\), we compute the Hamming weights of
all blocks of size \(B_j\).  Thus level \(0\) consists of the original input
bits, and level \(r\) consists of the single global Hamming weight.

We first describe the forward computation.  Suppose levels \(0,1,\ldots,j-1\)
have already been computed.  Consider one block of size \(B_j\).  It consists $\frac{B_j}{B_{j-1}}=n^{1/r}$ subblocks from level \(j-1\).  The Hamming weights of these subblocks are
already available.  Writing each such subblock weight in binary, the number of
weighted input bits needed to form the level-\(j\) block weight is
$ 
    M_j
    =
    O\!\left(n^{1/r}\log B_{j-1}\right)
    \le
    O\!\left(n^{1/r}\log n\right),
$ 
and the resulting weighted sum lies in the range \(0,\ldots,B_j\).

By Lemma~\ref{lem:weighted-counting}, one level-\(j\) block weight can be
computed in constant depth using
\[
    O\!\left(M_j B_j+B_j\,\mathrm{polylog}\,B_j\right)
    =
    O\!\left(n^{(j+1)/r}\,\mathrm{polylog}\,n\right)
\]
ancillary qubits.  There are \(n/B_j=n^{1-j/r}\) level-\(j\) blocks, and all
of them are disjoint, so they can be processed in parallel.  Therefore the
total workspace needed for level \(j\) is
\[
    n^{1-j/r}\cdot
    O\!\left(n^{(j+1)/r}\,\mathrm{polylog}\,n\right)
    =
    O\!\left(n^{1+1/r}\,\mathrm{polylog}\,n\right).
\]
The depth of one level is constant.

Repeating this for \(j=1,2,\ldots,r\), the forward computation produces the
global Hamming weight at the top level in depth \(O(r)\).  We then copy this
top-level value into the final output register using CNOT gates.  The
intermediate weight registers at all levels are left in place for the moment.

Finally, we reverse the computation level by level, from \(j=r\) down to
\(j=1\).  Since each level was computed reversibly from the previous level,
this uncomputes all intermediate Hamming-weight registers and restores all
ancillary workspace to \(\ket0\), while leaving the copied final output
register unchanged.  The reverse pass has the same depth \(O(r)\).

Hence the total depth is \(O(r)\).  The ancillary cost is bounded by summing
the workspace bound over the \(r\) levels, together with the intermediate
weight registers.  This gives $O\!\left(r\,n^{1+1/r}\,\mathrm{polylog}\,n\right)$
additional clean ancillary qubits.
\end{proof}

This follows immediately from Theorem~\ref{thm:alltoall-dynamic-recursive} for a fixed $r>0$.
\begin{corollary}
\label{cor:alltoall-dynamic-fixed-r}
For every constant integer \(r\ge1\), $\HWC_n$ can be
implemented in the dynamic all-to-all model with depth \(O(1)\) and
$O\!\left(n^{1+1/r}\,\mathrm{polylog}\,n\right)$
additional clean ancillary qubits.
\end{corollary}

\begin{corollary}[Near-linear ancilla-count dynamic HWC]
\label{cor:alltoall-dynamic-near-linear}
For every fixed \(\varepsilon>0\), \(\HWC_n\) admits a dynamic all-to-all
implementation with depth \(O_\varepsilon(1)\) and
\(O(n^{1+\varepsilon}\polylog n)\) additional clean ancillas.
\end{corollary}

\begin{proof}
Choose \(r=\max\{1,\lceil 1/\varepsilon\rceil\}\).  Then \(1/r\le\varepsilon\).
By Theorem~\ref{thm:alltoall-dynamic-recursive}, the depth is \(O(r)=
O_\varepsilon(1)\), and the ancillary cost is
$
    O\!\left(r\,n^{1+1/r}\,\mathrm{polylog}\,n\right)
    \le
    O\!\left(r\,n^{1+\varepsilon}\,\mathrm{polylog}\,n\right).
$
Since \(\varepsilon\) is fixed, the factor \(r=O(1/\varepsilon)\) is absorbed
into the constant hidden in the \(O(\cdot)\) notation.
\end{proof}

\section{Two-Dimensional Nearest-Neighbor Constructions for Hamming Weight Computation}
\label{sec:two-dimensional}

We now turn to circuits in the two-dimensional nearest-neighbor square-grid (2D) qubit connectivity.  
Throughout this section, we first study the standard 2D model and show that \(\HWC_n\) can be implemented with
optimal depth \(O(\sqrt n)\) using only \(O(\log^2 n)\) additional clean
ancillas in Section \ref{sec:hwc-standard-2D}.  The dynamic 2D model will be considered in Section \ref{sec:hwc-dynamic-2D}.

\subsection{Hamming Weight Computation in Standard Circuit Model}
\label{sec:hwc-standard-2D}

A black-box routing argument would simulate an all-to-all depth-\(D\) circuit
on \(K\) qubits by a two-dimensional nearest-neighbor circuit of depth
\(O(\sqrt K D)\), using SWAP routing.  Instead, we use the special
Fourier phase-encoding structure of the ancilla-free Hamming-weight
construction in Lemma~\ref{lem:prelim-hw-tcad}.  The key step is to implement
phase gadgets of the form
$
    \exp\!\left(i\sum_j \theta_j y_j\sum_i x_i\right)
$
directly on a two-dimensional patch.

\begin{lemma}[Dirty fan-out on a bounded-degree tree]
\label{lem:dirty-fanout-tree}
Let \(T=(V,E)\) be a rooted tree of height \(h\) and constant maximum degree, with
root \(\rho\).  The dirty fan-out operation
\[
    \ket{a_\rho}\bigotimes_{v\in V\backslash\{\rho\}}\ket{a_v}
    \longmapsto
    \ket{a_\rho}\bigotimes_{v\in V\backslash\{\rho\}}\ket{a_v\oplus a_\rho},
\]
can be implemented by CNOT gates along the edges of \(T\), with depth \(O(h)\)
and without clean ancillary qubits.
\end{lemma}

\begin{proof}
Let \(p(v)\) denote the parent of a non-root vertex \(v\).  We use two sweeps
of CNOT gates, always directed from parent to child.

First sweep from the leaves toward the root.  For depths
\(d=h-1,h-2,\ldots,0\), apply CNOTs from every vertex at depth \(d\) to its
children at depth \(d+1\).  Since the maximum degree is constant, the CNOTs
between two consecutive levels can be scheduled in \(O(1)\) layers.  After
this sweep, every non-root vertex \(v\) stores \(a_v\oplus a_{p(v)}\), because
the parent \(p(v)\) has not yet been modified by its own parent when $\CNOT^{p(v)}_v$ is applied.

Second sweep from the root toward the leaves, but skip the edges from the root
to depth \(1\).  Thus, for \(d=1,2,\ldots,h-1\), apply CNOTs from every vertex
at depth \(d\) to its children at depth \(d+1\).  We claim that after this
second sweep reaches level \(d\), every vertex \(v\) of depth at most \(d\)
stores \(a_v\oplus a_\rho\).  This is true at depth \(1\) after the first
sweep.  If a parent \(p\) stores \(a_p\oplus a_\rho\) and its child \(c\)
stores \(a_c\oplus a_p\), then $\CNOT^p_c$ changes the child to
\[
    (a_c\oplus a_p)\oplus(a_p\oplus a_\rho)
    =
    a_c\oplus a_\rho .
\]
Thus the claim propagates to all levels.  Both sweeps have depth \(O(h)\).
\end{proof}

\begin{corollary}[Dirty fan-out on a two-dimensional patch]
\label{cor:dirty-fanout-2d}
On a two-dimensional square-grid patch of area \(O(S)\), dirty fan-out from
one root qubit to \(S\) target qubits can be implemented with depth
\(O(\sqrt S)\) and without clean ancillary qubits.
\end{corollary}

\begin{proof}
Choose a breadth-first spanning tree of the patch rooted at the root qubit.
The tree has height \(O(\sqrt S)\), and the square-grid has constant degree.
The claim follows from Lemma~\ref{lem:dirty-fanout-tree}.
\end{proof}

\begin{lemma}
\label{lem:single-output-phase-gadget-2d}
Let \(x_1,\ldots,x_S \in\{0,1\}\) be data qubits and let \(y\) be a control qubit, all
placed in a two-dimensional patch of area \(O(S)\).  For any angle \(\theta \in\mathbb{R}\),
the diagonal unitary
\[
    \ket{x_1,\ldots,x_S}\ket y
    \longmapsto
    e^{i\theta y\sum_{i=1}^S x_i}
    \ket{x_1,\ldots,x_S}\ket y, \quad \forall x_1,\ldots,x_S,y\in\{0,1\}
\]
can be implemented with depth \(O(\sqrt S)\), using no additional clean
ancillary qubits.
\end{lemma}

\begin{proof}
Let \(P(\alpha)=\operatorname{diag}(1,e^{i\alpha})\).  First apply dirty
fan-out from \(y\) to \(x_1,\ldots,x_S\).  By
Corollary~\ref{cor:dirty-fanout-2d}, this has depth \(O(\sqrt S)\) and maps
\(x_i\mapsto x_i\oplus y\) for every \(i\), while leaving \(y\) unchanged.

Next apply \(P(-\theta/2)\) to all \(x_i\)'s in parallel.  On a computational
basis state, this contributes
\[
    \exp\!\left(-\frac{i\theta}{2}\sum_{i=1}^S(x_i\oplus y)\right)=\exp\!\left(-\frac{i\theta}{2}\sum_i x_i\right)
    \exp\!\left(-\frac{i\theta}{2}S y\right)
    \exp\!\left(i\theta y\sum_i x_i\right),
\]
since \(x_i\oplus y=x_i+y-2x_i y\).
The last factor is the desired phase.  We then reverse the dirty fan-out and
cancel the two extra factors by applying \(P(\theta/2)\) to every \(x_i\) and
\(P(S\theta/2)\) to \(y\).  The only nonconstant-depth operations are the
dirty fan-out and its inverse, so the depth is \(O(\sqrt S)\).
\end{proof}

\begin{lemma}[Cyclic tour phase on a corridor]
\label{lem:cyclic-tour-corridor}
Consider a cycle of qubit positions containing \(K\) control tokens
\(Y_1,\ldots,Y_K\) and \(N_{\mathrm{corr}}\) data tokens.  Suppose that,
initially, between every two consecutive control tokens there are exactly
\(q\) data tokens, so \(N_{\mathrm{corr}}=Kq\).  Let the data tokens carry bits
\(x\), and let \(Y_j\) carry the bit \(y_j\).  Then, for fixed angles
\(\theta_1,\ldots,\theta_K\), one can implement the corridor phase
\[
    \exp\!\left(
        i\sum_{j=1}^K \theta_j y_j
        \sum_{x\in \mathrm{corridor}} x
    \right)
\]
using nearest-neighbor gates along the cycle, with depth \(O(K(q+1))\), while
restoring every token to its original position.  No clean ancillary qubits are
used.
\end{lemma}

\begin{proof}
We use a cyclic token-swapping schedule.  Figure~\ref{fig:cyclic-tour-example}
illustrates one elementary step for the example \(K=3\) and \(q=2\).  The
controls are initially separated by two data tokens.  In one elementary step,
each control crosses the next data token in the clockwise direction, and the
corresponding controlled phase is applied before the swap.

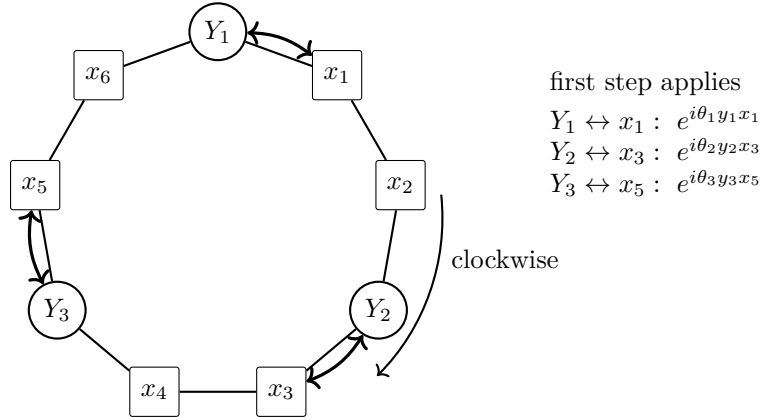
\begin{figure}[t]
\centering
\begin{tikzpicture}[
    scale=1.0,
    every node/.style={font=\small},
    ctl/.style={draw,circle,thick,minimum size=7.5mm,inner sep=0pt},
    dat/.style={draw,rectangle,rounded corners=1pt,minimum size=6.5mm,inner sep=1pt},
    swap/.style={<->,very thick,bend left=18}
]
\node[ctl] (Yone)   at (90:2.45)  {$Y_1$};
\node[dat] (xone)   at (50:2.45)  {$x_1$};
\node[dat] (xtwo)   at (10:2.45)  {$x_2$};
\node[ctl] (Ytwo)   at (-30:2.45) {$Y_2$};
\node[dat] (xthree) at (-70:2.45) {$x_3$};
\node[dat] (xfour)  at (-110:2.45){$x_4$};
\node[ctl] (Ythree) at (-150:2.45){$Y_3$};
\node[dat] (xfive)  at (170:2.45) {$x_5$};
\node[dat] (xsix)   at (130:2.45) {$x_6$};

\draw[thick]
    (Yone) -- (xone) -- (xtwo) -- (Ytwo) -- (xthree) -- (xfour)
    -- (Ythree) -- (xfive) -- (xsix) -- (Yone);

\draw[swap] (Yone) to (xone);
\draw[swap] (Ytwo) to (xthree);
\draw[swap] (Ythree) to (xfive);

\draw[->,thick] (2.95,0.30) arc[start angle=6,end angle=-45,radius=2.95];
\node[anchor=west] at (2.95,-0.55) {clockwise};

\node[anchor=west,align=left] at (4.25,1.10) {
first step applies\\[0.25em]
\(Y_1\leftrightarrow x_1:\ e^{i\theta_1 y_1x_1}\)\\
\(Y_2\leftrightarrow x_3:\ e^{i\theta_2 y_2x_3}\)\\
\(Y_3\leftrightarrow x_5:\ e^{i\theta_3 y_3x_5}\)
};
\end{tikzpicture}
\caption{One elementary step of the cyclic-tour schedule for \(K=3\) and
\(q=2\).  The double arrows indicate the three disjoint control--data swaps
performed in parallel.  Before each swap, the corresponding controlled phase
is applied.}
\label{fig:cyclic-tour-example}
\end{figure}

Formally, at each elementary step, every control token \(Y_j\) is swapped with
the next data token in the clockwise direction.  Since there are always
\(q\) data tokens between two consecutive controls, the \(K\) swaps in one
elementary step are disjoint and can be performed in parallel.

Whenever \(Y_j\) crosses a data token carrying bit \(x\), we first apply the
two-qubit diagonal phase
\[
    \ket{y_j}\ket{x}
    \longmapsto
    e^{i\theta_j y_jx}\ket{y_j}\ket{x},\quad \forall y_j,x\in\{0,1\}
\]
and then perform the swap.  We apply these phases only during the first
\(N_{\mathrm{corr}}=Kq\) elementary steps.

The key point is that the cyclic order of the control tokens is preserved, and
the cyclic order of the data tokens is also preserved: controls never cross
other controls, and data tokens never cross other data tokens.  Therefore,
during the first \(N_{\mathrm{corr}}\) crossings, each control token crosses
every data token exactly once.  The accumulated phase is
\[
    \prod_{j=1}^K\prod_{x\in\mathrm{corridor}} e^{i\theta_j y_jx}
    =
    \exp\!\left(
        i\sum_{j=1}^K \theta_j y_j
        \sum_{x\in \mathrm{corridor}} x
    \right),
\]
which is the desired corridor phase.

After these \(N_{\mathrm{corr}}\) phase-accumulating steps, we continue the
same swap schedule, without applying any further phase gates, until each
control has completed one full tour through the cycle of root positions.  This
takes \(K(q+1)\) elementary steps in total.  We then reverse the entire swap
schedule, again without applying further phases.  This restores every token to
its original position.  The total depth is \(O(K(q+1))\), and no clean
ancillary qubits are used.
\end{proof}

\begin{lemma}[Multi-output phase gadget on a two-dimensional patch]
\label{lem:multi-output-phase-gadget-2d}
Let \(x_1,\ldots,x_B\) be data qubits and let \(y_1,\ldots,y_k\) be control
qubits, with \(1\le k\le B\).  Suppose all these qubits are placed in a
two-dimensional patch of area \(O(B)\).  For fixed angles
\(\theta_1,\ldots,\theta_k \in \mathbb{R}\), the diagonal unitary
\[
    \ket{x}\ket{y}
    \longmapsto
    \exp\!\left(
        i\sum_{j=1}^k \theta_j y_j\sum_{i=1}^B x_i
    \right)
    \ket{x}\ket{y}, \qquad \forall x=(x_1,\ldots,x_B), \quad \forall y=(y_1,\ldots,y_k),
\]
can be implemented with depth \(O(\sqrt{Bk})\) using nearest-neighbor gates
and \(O(k)\) clean padding qubits.
\end{lemma}

\begin{proof}
Let \(K=\Theta(k)\) be a convenient padded number of control roots.  Add
\(K-k\) dummy control qubits initialized to \(\ket0\), and extend the angle
list by setting \(\theta_j=0\) for \(k<j\le K\).  It is enough to implement
the corresponding \(K\)-control phase gadget.

Choose integers \(a,b=\Theta(\sqrt K)\) with \(ab=K\), so that the \(K\)
controls can be arranged on an \(a\times b\) coarse grid.  Set
$ 
    s=\left\lceil\sqrt{B/K}\right\rceil .
$
Partition the \(O(B)\)-area patch into \(K\) square subpatches of side
\(\Theta(s)\), arranged as this coarse \(a\times b\) grid.  The total
subpatch area is
$ 
    O(Ks^2)=O(B).
$ 
Choose one root site in each subpatch.

Regard the subpatches as the vertices of the coarse grid.  After a
constant-factor padding of the layout, choose a Hamiltonian cycle through the
\(K\) coarse-grid vertices.  Lifting each coarse edge to a grid path between
the corresponding root sites gives a corridor cycle through all \(K\) roots.
Each corridor segment has length \(O(s)\).  By padding each segment by a
constant factor if necessary, we may assume that between consecutive roots on
the corridor cycle there are exactly \(q=\Theta(s)\) corridor data qubits.  Hence the total corridor length is
$ 
    Kq=O(Ks)=O(\sqrt{BK}),
$ 
which is at most \(O(B)\) since \(K=\Theta(k)\le O(B)\).  Thus the entire
layout fits inside an \(O(B)\)-area patch.  Figure~\ref{fig:multi-output-phase-layout}
shows the structure of the construction.

\begin{figure}[t]
\centering
\begin{tikzpicture}[
    x=1cm,y=1cm,
    every node/.style={font=\small},
    subpatch/.style={draw,thick},
    root/.style={circle,fill=red,inner sep=1.6pt},
    interior/.style={circle,fill=gray!60,inner sep=0.7pt},
    corridor/.style={circle,fill=blue!60,inner sep=0.7pt},
    corridorline/.style={very thick,blue!70},
    >=Latex
]
\draw[thick] (0,0) rectangle (8.4,8.4);

\foreach \x in {0.4,3.0,5.6}{
  \foreach \y in {0.4,3.0,5.6}{
    \draw[subpatch] (\x,\y) rectangle ++(2.2,2.2);
  }
}

\coordinate (r1) at (1.5,6.7);
\coordinate (r2) at (4.1,6.7);
\coordinate (r3) at (6.7,6.7);
\coordinate (r4) at (6.7,4.1);
\coordinate (r5) at (4.1,4.1);
\coordinate (r6) at (1.5,4.1);
\coordinate (r7) at (1.5,1.5);
\coordinate (r8) at (4.1,1.5);
\coordinate (r9) at (6.7,1.5);

\foreach \p in {r1,r2,r3,r4,r5,r6,r7,r8,r9}{
  \node[root] at (\p) {};
}

\foreach \p in {(0.9,6.0),(2.0,7.4),(3.5,5.9),(4.8,7.3),(6.1,5.9),
                (0.9,3.4),(2.1,4.8),(3.5,3.3),(4.9,4.8),(6.1,3.3),
                (0.9,0.8),(2.1,2.2),(3.4,0.8),(4.9,2.2),(6.1,0.8)}{
  \node[interior] at \p {};
}

\draw[corridorline]
    (r1) -- (r2) -- (r3) -- (r4) -- (r5) -- (r6) -- (r7) -- (r8) -- (r9)
    -- (7.9,1.5) -- (7.9,7.9) -- (1.5,7.9) -- cycle;

\foreach \p in {(2.8,6.7),(5.4,6.7),(6.7,5.4),(5.4,4.1),(2.8,4.1),
                (1.5,2.8),(2.8,1.5),(5.4,1.5),(7.9,3.2),(7.9,5.0),
                (6.2,7.9),(3.8,7.9)}{
  \node[corridor] at \p {};
}

\draw[->,corridorline] (3.0,6.95) -- (3.6,6.95);
\draw[->,corridorline] (6.95,5.0) -- (6.95,4.4);
\draw[->,corridorline] (2.8,1.75) -- (3.4,1.75);

\node[anchor=west] at (8.7,6.9) {corridor cycle};
\draw[->] (8.55,6.8) -- (7.9,6.8);

\node[anchor=west] at (8.7,5.8) {root site};
\draw[->] (8.55,5.7) -- (6.7,6.7);

\node[anchor=west] at (8.7,4.6) {corridor data};
\draw[->] (8.55,4.5) -- (7.9,5.0);

\node[anchor=west] at (8.7,3.3) {subpatch interior data};
\draw[->] (8.55,3.2) -- (4.9,4.8);

\node[anchor=west] at (8.7,2.1) {one round: interior phase,};
\node[anchor=west] at (8.7,1.6) {then move to next root.};

\node at (4.2,-0.45) {schematic layout of the \(K\)-control gadget};
\end{tikzpicture}
\caption{Layout for the multi-output phase gadget.  The large patch is
partitioned into \(K\) square subpatches.  Each subpatch contains one root
site.  The corridor cycle passes through all roots.  In each round, the
control currently at a root first interacts with the interior data of that
subpatch, and then moves along the corridor to the next root.}
\label{fig:multi-output-phase-layout}
\end{figure}

Place the \(K\) real or dummy controls at the \(K\) roots.  The remaining
qubits are data qubits, divided into corridor data qubits and subpatch-interior
data qubits.

We perform \(K\) rounds.  At the beginning of each round, every control qubit
is located at the root of some subpatch.  In parallel over all subpatches,
apply Lemma~\ref{lem:single-output-phase-gadget-2d} between the control at the
root and all interior data qubits of that subpatch.  Since each subpatch has
area \(O(s^2)\), this interior phase step has depth \(O(s)\).  If \(y_j\) is
currently at the root of subpatch \(T\), the contributed phase is
\[
    \exp\!\left(
        i\theta_j y_j \sum_{x_i\in T^\circ} x_i
    \right),
\]
where \(T^\circ\) denotes the interior of \(T\).

After the interior phase step, move all controls to the next roots along the
corridor cycle.  This move has depth \(O(q)=O(s)\).  During the move, apply
the corridor phases from Lemma~\ref{lem:cyclic-tour-corridor}.  Over the full
tour, the corridor data contribute the phase
\[
    \exp\!\left(
        i\sum_{j=1}^K \theta_j y_j
        \sum_{x_i\in\mathrm{corridor}} x_i
    \right).
\]

After \(K\) rounds, each real control has visited every subpatch root exactly
once.  Therefore the interior data contribute
\[
    \exp\!\left(
        i\sum_{j=1}^K \theta_j y_j
        \sum_{x_i\in\mathrm{interiors}} x_i
    \right).
\]
Multiplying this with the corridor contribution gives
\[
    \exp\!\left(
        i\sum_{j=1}^k \theta_j y_j\sum_{i=1}^B x_i
    \right),
\]
because the dummy controls satisfy \(\theta_j=0\) for \(j>k\).

Finally, reverse the corridor token-swapping schedule so that all controls and
corridor data return to their original positions.  The single-output phase
gadgets are clean because their internal dirty fan-out is uncomputed within
each application.  Each round has depth \(O(s)\), and there are \(K=\Theta(k)\)
rounds, so the total depth is
\[
    O(Ks)
    =
    O\!\left(K\sqrt{B/K}\right)
    =
    O(\sqrt{Bk}).
\]
The only additional clean ancillary qubits are the \(O(k)\) dummy controls.
\end{proof}

\begin{remark}
The two geometric lemmas above can be strengthened slightly, but these
refinements do not affect the asymptotic bounds used later.

First, in Lemma~\ref{lem:cyclic-tour-corridor}, the assumption that every gap
contains exactly \(q\) data tokens is only for convenience.  The same
token-swapping argument works when the \(j\)-th gap contains \(q_j\) data
tokens, as long as the total corridor length is controlled; one then obtains
depth \(O\!\left(\sum_j(q_j+1)\right)\).

Second, in Lemma~\ref{lem:multi-output-phase-gadget-2d}, the \(O(k)\) clean
ancillas come only from the dummy controls used to make the coarse layout more
regular.  With a less uniform layout and the variable-gap version of
Lemma~\ref{lem:cyclic-tour-corridor}, one can remove these extra ancillas.
We keep the present formulation because, in the later application to
Hamming-weight computation, we use \(k=O(\log B)\), so these \(O(k)\) extra
clean qubits do not change the final asymptotic bounds.
\end{remark}

\begin{lemma}
\label{lem:hw-tcad-2d-block}
The \(B\)-bit Hamming weight computation \(\HWC_B\) can be implemented on a
two-dimensional nearest-neighbor patch of area \(O(B)\) with depth
$ 
    O(\sqrt{B\log B})
$ 
using \(O(\log B)\) additional clean ancillary qubits beyond the local output
register.
\end{lemma}

\begin{proof}
Use the structure of the construction behind Lemma~\ref{lem:prelim-hw-tcad}.
That construction first encodes \(|x|=\sum_{i=1}^B x_i\) into phases of the
\(m_B=\lceil\log(B+1)\rceil\)-qubit output register, and then applies an
inverse QFT on that output register.

The phase-encoding part consists of \(m_B=O(\log B)\) phase interactions of
the form
\[
    \exp\!\left(i\theta_j y_j\sum_{i=1}^B x_i\right),
\]
one for each output qubit \(y_j\).  Applying
Lemma~\ref{lem:multi-output-phase-gadget-2d} with \(k=m_B=O(\log B)\)
implements all these phase interactions together in depth
$ 
    O(\sqrt{B m_B})=O(\sqrt{B\log B}),
$ 
using \(O(m_B)=O(\log B)\) clean padding qubits.

After the phase encoding, the inverse QFT acts only on the \(m_B=O(\log B)\)
output qubits.  Route these output qubits to a small region, apply the inverse
QFT, and then reverse the routing.  The routing cost is at most the diameter
of the patch, \(O(\sqrt B)\), and the inverse QFT itself has depth
\(O(m_B^2)=O(\log^2 B)\) on a nearest-neighbor layout of \(m_B\) qubits.  Both
costs are absorbed into \(O(\sqrt{B\log B})\) for asymptotic \(B\).  Hence the
block computation has depth \(O(\sqrt{B\log B})\) and uses \(O(\log B)\)
additional clean ancillary qubits beyond the local output register.
\end{proof}

\begin{theorem}[Depth-optimal HWC on 2D grid]
\label{thm:hw-2d-polylog-ancilla}
In the standard 2D model, \(\HWC_n\)
admits an implementation with depth $O(\sqrt n)$ and $O(\log^2 n)$ additional clean ancillary qubits.
Moreover, the depth \(O(\sqrt n)\) is asymptotically optimal.
\end{theorem}

\begin{proof}
Let \(B=\Theta(n/\log n)\).  Partition the input into
\(L=\lceil n/B\rceil\) blocks of size at most \(B\).  Place each block inside
a two-dimensional patch of area \(O(B)\), and arrange all block patches inside
the ambient \(O(n)\)-area grid.

For each block, compute its local Hamming weight using
Lemma~\ref{lem:hw-tcad-2d-block}.  All blocks are processed in parallel, and
the depth is
$ 
    O(\sqrt{B\log B})=O(\sqrt n).
$ 
The total number of local-weight qubits is
$ 
    M
    =
    L\lceil\log(B+1)\rceil
    =
    O\!\left(\frac nB\log B\right)
    =
    O(\log^2 n).
$ 
The padding qubits used inside the block gadgets contribute another
\(L\cdot O(\log B)=O(\log^2 n)\) clean qubits.

It remains to add the local weights into the final output register.  Write
the local weight of block \(\ell\) as \(w_\ell=\sum_j2^j b_{\ell,j}\),  then
$ 
    |x|=\sum_{\ell,j}2^j b_{\ell,j}.
$ 
Route the \(M=O(\log^2 n)\) local-weight qubits to a region near the final
output register.  Since the computation occupies an \(O(n)\)-area grid, this
gathering takes depth \(O(\sqrt n)\).  Inside a patch of area \(O(M)\), add
these weighted bits into the final output register using the carry-save
summation from Lemma~\ref{lem:prelim-carry-save}.  Even with nearest-neighbor
routing inside this \(O(M)\)-area patch, the summation depth is
\(\mathrm{polylog}\,n=o(\sqrt n)\), and the clean workspace is \(O(M)\).

After the global Hamming weight has been written to the final output register,
reverse the summation workspace, route the local-weight registers back to their
block patches in depth \(O(\sqrt n)\), and reverse all local Hamming-weight
computations.  This restores all intermediate registers to \(\ket0\), while
leaving the final output register in the state \(\ket{|x|}\).  The total depth
is \(O(\sqrt n)\), and the total additional clean ancillary cost is
\(O(\log^2 n)\).

For optimality, observe that the least significant output bit equals
\[
    |x|\bmod 2=x_1\oplus x_2\oplus\cdots\oplus x_n .
\]
In a depth-\(d\) nearest-neighbor circuit on a two-dimensional grid, the
backward light cone of any output qubit is contained in a radius-\(d\) region
and therefore contains at most \(O(d^2)\) input qubits.  Since the least
significant output bit depends on all \(n\) inputs, we must have
\(O(d^2)\ge n\), hence \(d=\Omega(\sqrt n)\).  The upper bound is therefore
depth-optimal up to constant factors.
\end{proof}

\subsection{Hamming Weight Computation in Dynamic Circuit Model}
\label{sec:hwc-dynamic-2D}

We now consider the dynamic 2D model.  
%
As usual, each clean ancillary qubit occupies one additional grid site.
Therefore, a construction using $A(n)$ clean ancillary qubits is laid out on $O(n+A(n))$ grid sites.


\begin{lemma}\label{lem:2d-dynamic-fanout-path}
    Given an input state of the form $\alpha\ket{0}\ket{0^r}_{\rm  leaf}\ket{\eta_0}+\beta\ket{1}\ket{1^r}_{\rm leaf}\ket{\eta_1}$, the quantum state $\alpha\ket{0}\ket{\eta_0}+\beta\ket{1}\ket{\eta_1}$ can be obtained by a constant-depth dynamic circuit without ancillary qubits using only single-qubit gates. 
\end{lemma}
\begin{proof}
First, we apply $r$ Hadamard gates $H^{\otimes r}$ to $r$ leaf qubits of the input state and obtain 
\(\alpha\ket 0 \ket{+}_{\rm leaf}^{\otimes r}\ket{\eta_0}+\beta\ket1\ket{-}_{\rm leaf}^{\otimes r}\ket{\eta_1}\).
Second, measure the $r$ leaf qubits in the computational basis and obtain the measurement outcome $c_1,\ldots,c_r\in\{0,1\}$. The
state after measurement is therefore
$\alpha\ket 0 \ket{\eta_0}+(-1)^{\oplus_{j=1}^r c_j}\beta\ket1\ket{\eta_1}$.
The leaf qubits are in state $\ket{c_1,\ldots,c_r}_{\rm leaf}$ and can be restored by $\bigotimes_{j=1}^r X^{c_j}$. Classically compute $\bigoplus_{j=1}^r c_j$ and apply $Z^{\bigoplus_{j=1}^r c_j}$ to the first qubit. This completes the proof.
\end{proof}
The dynamic circuit in Lemma \ref{lem:2d-dynamic-fanout-path} can be realized under any qubit connectivity since it consists of single-qubit gates.

We first adapt the weighted-counting primitive to the two-dimensional dynamic
layout.  We separate the construction into two independent parts: the
weighted Fourier phase preparation, and the Takahashi--Tani parallel decoder.

\begin{lemma}[Two-dimensional weighted phase preparation]
\label{lem:2d-dynamic-weighted-phase-preparation}
Let \(z_1,\ldots,z_M\in\{0,1\}\) be input bits with fixed nonnegative integer
weights \(a_1,\ldots,a_M\), and let
$
    S=\sum_{i=1}^M a_i z_i
$ 
satisfy \(0\le S\le R\).  Let \(q=\lceil\log(R+1)\rceil\).  For every branch indexed by $(k,y)\in \mathcal{I}=\{(k,y): 0\le k <q,y\in\{0,1\}^k\}$,
define
\[
    \ket{\phi_{k,y}(S)}
    =
    \frac{\ket0+e^{2\pi iS/2^{k+1}}\ket1}{\sqrt2}.
\]
The family of phase states \(\{\ket{\phi_{k,y}(S)}: (k,y)\in\mathcal{I}\}\), can be prepared in constant depth in dynamic
2D model using \(O(MR)\) additional clean ancillary qubits.
\end{lemma}

\begin{proof}
The number of branches is
$ 
    N_{\mathrm{br}}
    =
    \sum_{k=0}^{q-1}2^k
    =
    2^q-1
    =
    O(R).
$ 
The layout is shown schematically in Fig.~\ref{fig:2d-weighted-phase-prep}.
We use an \(M\times N_{\mathrm{br}}\) rectangular array of constant-size
cells.  A cell \((i,b)\), where \(b\) denotes a branch, contains two qubits:
a copy qubit \(X_{i,b}\), which will store a copy of \(z_i\), and a leaf qubit
\(L_{i,b}\), which will be part of the cat state for branch \(b\).  Since each
cell has constant area, the array has total area
$
    O(MN_{\mathrm{br}})=O(MR).
$

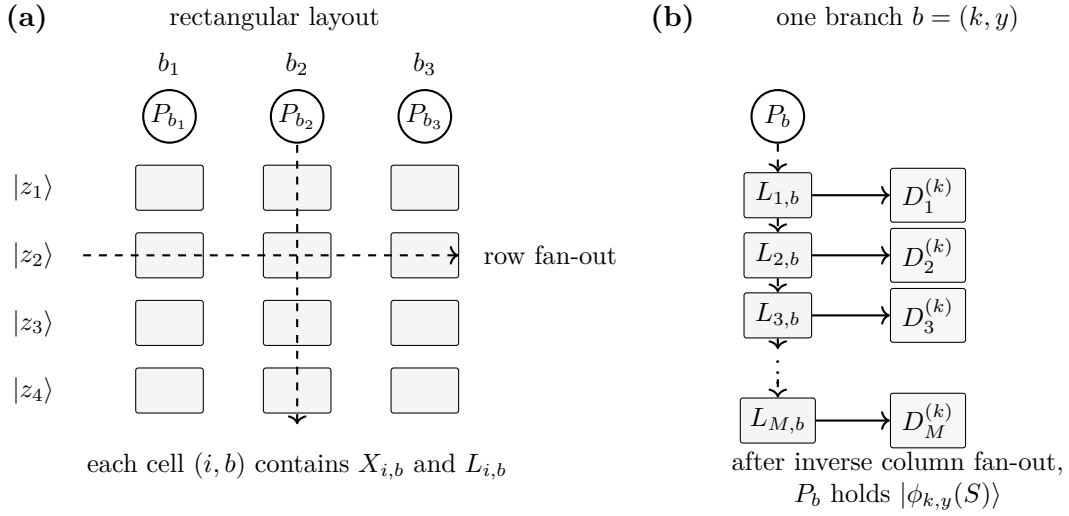
\begin{figure}[htpb]
\centering
\resizebox{0.95\textwidth}{!}{%
\begin{tikzpicture}[
    x=1cm,y=1cm,
    every node/.style={font=\small},
    cell/.style={draw,rounded corners=1pt,minimum width=9mm,minimum height=6mm,fill=gray!8},
    root/.style={draw,circle,thick,minimum size=7mm,inner sep=0pt},
    gate/.style={draw,rounded corners=1pt,minimum width=10mm,minimum height=5.5mm,fill=gray!6},
    arr/.style={->,thick},
    dasharr/.style={->,thick,dashed}
]


\node[font=\bfseries] at (0.6,0.9) {(a)};
\node at (3.9,0.9) {rectangular layout};

\node at (2.5,0.30) {$b_1$};
\node at (4.2,0.30) {$b_2$};
\node at (5.9,0.30) {$b_3$};

\node[root] (Pb1) at (2.5,-0.40) {$P_{b_1}$};
\node[root] (Pb2) at (4.2,-0.40) {$P_{b_2}$};
\node[root] (Pb3) at (5.9,-0.40) {$P_{b_3}$};

\node[left] at (1.10,-1.35) {$\ket{z_1}$};
\node[left] at (1.10,-2.25) {$\ket{z_2}$};
\node[left] at (1.10,-3.15) {$\ket{z_3}$};
\node[left] at (1.10,-4.05) {$\ket{z_4}$};

\node[cell] at (2.5,-1.35) {};
\node[cell] at (4.2,-1.35) {};
\node[cell] at (5.9,-1.35) {};

\node[cell] at (2.5,-2.25) {};
\node[cell] at (4.2,-2.25) {};
\node[cell] at (5.9,-2.25) {};

\node[cell] at (2.5,-3.15) {};
\node[cell] at (4.2,-3.15) {};
\node[cell] at (5.9,-3.15) {};

\node[cell] at (2.5,-4.05) {};
\node[cell] at (4.2,-4.05) {};
\node[cell] at (5.9,-4.05) {};

\draw[dasharr] (1.35,-2.25) -- (6.35,-2.25);
\node[anchor=west] at (6.55,-2.25) {row fan-out};

\draw[dasharr] (4.2,-0.78) -- (4.2,-4.50);

\node[align=center] at (4.2,-5.05)
{each cell \((i,b)\) contains \(X_{i,b}\) and \(L_{i,b}\)};


\node[font=\bfseries] at (9.2,0.9) {(b)};
\node at (12.2,0.9) {one branch \(b=(k,y)\)};

\node[root] (P) at (10.6,-0.40) {$P_b$};

\node[cell] (L1) at (10.6,-1.45) {$L_{1,b}$};
\node[cell] (L2) at (10.6,-2.25) {$L_{2,b}$};
\node[cell] (L3) at (10.6,-3.05) {$L_{3,b}$};
\node at (10.6,-3.70) {$\vdots$};
\node[cell] (LM) at (10.6,-4.45) {$L_{M,b}$};

\draw[dasharr] (P.south) -- (L1.north);
\draw[dasharr] (L1.south) -- (L2.north);
\draw[dasharr] (L2.south) -- (L3.north);
\draw[dasharr] (L3.south) -- (10.6,-3.50);
\draw[dasharr] (10.6,-3.90) -- (LM.north);

\node[gate] (D1) at (12.6,-1.45) {$D_1^{(k)}$};
\node[gate] (D2) at (12.6,-2.25) {$D_2^{(k)}$};
\node[gate] (D3) at (12.6,-3.05) {$D_3^{(k)}$};
\node[gate] (DM) at (12.6,-4.45) {$D_M^{(k)}$};

\draw[arr] (L1.east) -- (D1.west);
\draw[arr] (L2.east) -- (D2.west);
\draw[arr] (L3.east) -- (D3.west);
\draw[arr] (LM.east) -- (DM.west);

\node[align=center] at (12.2,-5.25)
{after inverse column fan-out,\\
\(P_b\) holds \(\ket{\phi_{k,y}(S)}\)};

\end{tikzpicture}%
}

\caption{Two-dimensional weighted phase preparation.  Panel~(a) shows the
\(M\times N_{\mathrm{br}}\) rectangular layout.  Row fan-out distributes each
input bit \(z_i\) to all branches, and column fan-out prepares one cat state
for each branch root \(P_b\).  Each cell \((i,b)\) contains a copy qubit
\(X_{i,b}\) and a leaf qubit \(L_{i,b}\).  Panel~(b) shows one branch
\(b=(k,y)\): the local gates \(D_i^{(k)}\) accumulate the weighted phase, and
after inverse column fan-out the root qubit is left in the state
\(\ket{\phi_{k,y}(S)}\).}
\label{fig:2d-weighted-phase-prep}
\end{figure}

We first distribute the input bits to all branches.  For each \(i\), use the
measurement-based path fan-out of Lemma~\ref{lem:fanout-path}
along the \(i\)-th row to implement
\[
    \ket{z_i}\bigotimes_{b=1}^{N_{\mathrm{br}}}\ket{0}_{X_{i,b}}
    \longmapsto
    \ket{z_i}\bigotimes_{b=1}^{N_{\mathrm{br}}}\ket{z_i}_{X_{i,b}} ,
\]
using $O(N_{\rm br})$ clean ancillary qubits.
All rows are disjoint, so these row fan-outs are performed in parallel, and have constant depth using $O(MN_{\rm br})$ ancillary qubits in total.

Next, for each branch \(b\), introduce a root qubit \(P_b\) initialized to
\(\ket{+}\).  Along the column corresponding to \(b\), apply the same
measurement-based path fan-out from \(P_b\) to the leaf qubits
\(L_{1,b},\ldots,L_{M,b}\).  By linearity, this maps
\[
    \ket{+}_{P_b}\ket{0^M}_{L_{1,b},\ldots,L_{M,b}}
    \longmapsto
    \frac{
        \ket{0}_{P_b}\ket{0^M}_{L_{1,b},\ldots,L_{M,b}}
        +
        \ket{1}_{P_b}\ket{1^M}_{L_{1,b},\ldots,L_{M,b}}
    }{\sqrt2},
\]
using $O(M)$ clean ancillary qubits. All columns are disjoint, so the cat states for all branches are prepared in
parallel and in constant depth. The total ancillary count is $O(MN_{\rm br})$.

Now fix one branch \(b=(k,y)\), and define
$ 
    \omega_k=e^{2\pi i/2^{k+1}}.
$
For every cell \((i,b)\), apply the local two-qubit diagonal gate
\[D_i^{(k)}:
    \ket{p}_{X_{i,b}}\ket{q}_{L_{i,b}}
    \longmapsto
    \omega_k^{a_i p q}
    \ket{p}_{X_{i,b}}\ket{q}_{L_{i,b}},\quad\forall p,q\in\{0,1\}\]

This gate is local inside the cell, and all such gates are applied in one
parallel layer.

Since the qubit $X_{i,b}$ is in state $\ket{z_i}$, the all-zero component of the branch cat state
accumulates no phase.  On the all-one component, where every qubit $L_{i,b}$ is in state $\ket{1}$,
the accumulated phase is
\[
    \prod_{i=1}^M \omega_k^{a_i z_i}
    =
    \omega_k^{\sum_i a_i z_i}
    =
    e^{2\pi iS/2^{k+1}}.
\]
Thus the state of branch \(b=(k,y)\) becomes
\[
    \frac{
        \ket{0}_{P_b}\ket{0^M}_{L_{1,b},\ldots,L_{M,b}}
        +
        e^{2\pi iS/2^{k+1}}
        \ket{1}_{P_b}\ket{1^M}_{L_{1,b},\ldots,L_{M,b}}
    }{\sqrt2}.
\]

Finally, apply the inverse column fan-out.  This restores the leaf qubits
\(L_{1,b},\ldots,L_{M,b}\) to \(\ket{0}\) and leaves the root qubit $P_b$ in the
phase state
\[
    \ket{\phi_{k,y}(S)}
    =
    \frac{\ket{0}+e^{2\pi iS/2^{k+1}}\ket{1}}{\sqrt2},
    \qquad \forall\, b=(k,y)\in\mathcal{I}.
\]
Doing this for all branches in parallel prepares the full redundant family of
phase states.  We then apply the inverse row fan-outs to restore every copy
qubit \(X_{i,b}\) to \(\ket{0}\).

All steps have constant depth.
The total number of clean
ancillary qubits is \(O(MN_{\mathrm{br}})=O(MR)\).
\end{proof}

In the following lemma, we use the exact constant-depth ${\rm OR}_s$ construction of
Takahashi and Tani~\cite[Sec.~3]{TakahashiTani2016Collapse}, but lay it out
directly on a two-dimensional dynamic grid.
\begin{lemma}[Multi-input AND on a two-dimensional dynamic patch]
\label{lem:2d-dynamic-multi-and}
Let \(s\ge2\).  The map
\[
    \ket{x_1,\ldots,x_s}\ket0_{\mathrm{out}}
    \longmapsto
    \ket{x_1,\ldots,x_s}
    \ket{x_1\wedge\cdots\wedge x_s}_{\mathrm{out}}
\]
can be implemented in constant depth in the dynamic
2D model, with
$ 
    O(s\log s)
$
additional clean ancillary qubits.
\end{lemma}

\begin{proof}
It is enough to implement \(\mathrm{OR}_s\), since
\[
    x_1\wedge\cdots\wedge x_s
    =
    1\oplus
    \mathrm{OR}_s(1\oplus x_1,\ldots,1\oplus x_s),
\]
and the extra \(X\) gates do not affect the asymptotic resources. Let
$
    m=\lceil\log(s+1)\rceil .
$ 
The construction of ${\rm OR}_s$ consists of two parts.

First, they use the Høyer--Špalek OR-reduction to reduce \(\mathrm{OR}_s\)
to \(\mathrm{OR}_m\).  This reduction is an exact constant-depth
unbounded-fan-out circuit of size \(O(sm)=O(s\log s)\).  We place the
workspace of this reduction on an \(s\times O(m)\) rectangular patch: the
\(i\)-th row contains the \(O(m)\) work qubits associated with the input
\(x_i\).  The fan-out operations appearing in the reduction are implemented
along rows, columns, or constant-width strips by the measurement-based path
fan-out primitive of Lemma~\ref{lem:fanout-path}.  Since the
reduction has only constant depth and all strips used in a layer are disjoint,
this stage still has constant depth and uses area \(O(sm)\).  It produces
\(m\) reduced bits \(r_1,\ldots,r_m\) such that
\[
    \mathrm{OR}_s(x_1,\ldots,x_s)
    =
    \mathrm{OR}_m(r_1,\ldots,r_m).
\]

Second, implement $\mathrm{OR}_m$ by expanding it in the
Fourier basis over parity functions.  There are \(2^m=O(s)\) parity terms,
one for each subset \(T\subseteq[m]\), and each term involves at most \(m\)
reduced bits.  We lay out this stage on an \(m\times 2^m\) rectangular patch.
The \(j\)-th row corresponds to the reduced bit \(r_j\), and the column
indexed by \(T\subseteq[m]\) computes the parity
\[
    \bigoplus_{j\in T} r_j .
\]
The reduced bits are copied across the rows by measurement-based path
fan-out.  Inside each column, the required parity is computed by the standard
fan-out--parity equivalence: conjugating fan-out by Hadamard gates on all
participating qubits gives parity~\cite{HoyerSpalek2005Fanout}.  The
coefficient-dependent one-qubit phase gates from the Fourier expansion are
then applied locally in each column, and the remaining output combination is
again performed using fan-out/parity along allocated rows or columns.  Thus
the whole small-\(\mathrm{OR}_m\) stage has constant depth and area
$ 
    O(m2^m)=O(s\log s).
$ 

Combining the two stages gives an exact constant-depth two-dimensional dynamic
implementation of \(\mathrm{OR}_s\) using \(O(s\log s)\) clean workspace.  To
obtain a clean output map, we compute the OR value into a temporary output
qubit, copy it to the designated output register, and then reverse the two
stages to restore all work qubits to \(\ket0\).  Applying the initial and
final \(X\) gates for De Morgan gives the desired AND map with the same
asymptotic depth and ancillary cost.
\end{proof}

In the following lemma, we lay out the same parallel decoder of Takahashi and Tani~\cite[Sec.~4.1 and App.~A.3]{TakahashiTani2016Collapse} in the dynamic
2D model.

To state the decoder as a clean subroutine, we use the following shorthand.
Let \(q=\lceil\log(R+1)\rceil\), and for \(0\le t\le R\) define
\[
    \ket{\phi_{k,y}(t)}
    =
    \frac{\ket0+e^{2\pi i t/2^{k+1}}\ket1}{\sqrt2},
    \qquad
    \ket{\Phi_R(t)}
    =
    \bigotimes_{k=0}^{q-1}
    \bigotimes_{y\in\{0,1\}^k}
    \ket{\phi_{k,y}(t)} .
\]

\begin{lemma}[Two-dimensional layout of the parallel decoder]
\label{lem:2d-dynamic-parallel-decoder}
For every \(0\le t\le R\), given the redundant Fourier encoding
\(\ket{\Phi_R(t)}\), the decoder map
\[
    \ket{\Phi_R(t)}
    \ket{0^q}_{\mathrm{out}}
    \longmapsto
    \ket{\Phi_R(t)}
    \ket{t}_{\mathrm{out}}
\]
can be implemented in constant depth in the dynamic
2D model using \(O(R\,\mathrm{polylog}\,R)\) additional clean ancillary
qubits.  Here \(\ket{t}_{\mathrm{out}}\) denotes
the \(q\)-bit binary representation of \(t\).
\end{lemma}

\begin{proof}

Write \(t=\sum_{j=0}^{q-1}t_j2^j\).  For each \(k\) and each lower-bit guess
\(y=(y_0,\ldots,y_{k-1})\in\{0,1\}^k\), the Takahashi--Tani decoder applies
the phase correction associated with \(y\) to one copy of
\(\ket{\phi_{k,y}(t)}\), and then extracts a candidate qubit \(t_k^y\).  If
\(y=(t_0,\ldots,t_{k-1})\), then \(t_k^y=t_k\); otherwise the candidate may
be arbitrary.

The wrong guesses are removed by a prefix-consistency test.  For
\(y=(y_0,\ldots,y_{k-1})\), let \(y_{<j}=(y_0,\ldots,y_{j-1})\).  Define
\[
    e_j(y)=1\oplus y_j\oplus t_j^{y_{<j}},
    \qquad
    u_k(y)=t_k^y\wedge\bigwedge_{j=0}^{k-1}e_j(y).
\]
Here \(e_j(y)\) tests whether the guessed bit \(y_j\) agrees with the
candidate obtained from the prefix \(y_{<j}\).  If \(y\) is the correct
lower-bit string, then every prefix is correct, so all \(e_j(y)=1\), and
\(u_k(y)=t_k\).  If \(y\) is incorrect, let \(j\) be the first position where
\(y_j\ne t_j\).  Then \(y_{<j}\) is correct, hence \(t_j^{y_{<j}}=t_j\), so
\(e_j(y)=1\oplus y_j\oplus t_j=0\).  Therefore all incorrect guesses give
\(u_k(y)=0\), and
\[
    t_k=\bigoplus_{y\in\{0,1\}^k}u_k(y).
\]

We now describe the layout.  For each selector \(u_k(y)\), allocate a local
two-dimensional patch \(B_{k,y}\).  The patches are arranged in horizontal
bands indexed by \(k\).  Band \(k\) contains the \(2^k\) patches
\(B_{k,y}\), ordered lexicographically by \(y\in\{0,1\}^k\).  Each selector is
an AND of at most \(q+1=O(\log R)\) computational-basis bits, so by
Lemma~\ref{lem:2d-dynamic-multi-and}, each patch \(B_{k,y}\) has area
\(O(q\log q)\) and constant depth.  The total selector-patch area is therefore
\[
    \sum_{k=0}^{q-1}2^k\cdot O(q\log q)
    =
    O(Rq\log q).
\]

It remains to supply the candidate bits to the selector patches.  The selector
\(u_k(y)\) needs the candidate bits
\[
    t_k^y,\quad
    t_0^\emptyset,\quad
    t_1^{y_0},\quad
    \ldots,\quad
    t_{k-1}^{y_{<k-1}} .
\]
For a fixed candidate bit \(t_j^u\), the selectors in band \(k>j\) that need
this bit are exactly those whose string \(y\) has prefix \(u\).  Since the
patches in each band are ordered lexicographically, these selectors form a
contiguous interval.  We distribute \(t_j^u\) to the corresponding input ports
using a measurement-based fan-out strip running alongside that interval.
Different prefixes of the same length give disjoint intervals, and different
prefix lengths use separate strips.

With this band layout, the total number of candidate occurrences is
\[
    \sum_{k=0}^{q-1}2^k(k+1)=O(Rq).
\]
A straightforward allocation of \(O(q)\) strip tracks per band, together with
the selector patches of width \(O(q\log q)\), uses total area
\(O(Rq^2\log q)\).  Since \(q=\lceil\log(R+1)\rceil\), this is
\(O(R\,\mathrm{polylog}\,R)\).  All bands and all selector patches are
disjoint, so the selectors are evaluated in parallel and in constant depth.

After the selector terms are computed, for each fixed \(k\) we compute
\(\bigoplus_{y\in\{0,1\}^k}u_k(y)\) along band \(k\).  We use the standard
fan-out--parity equivalence: conjugating a fan-out operation by Hadamard gates
on all participating qubits gives the corresponding parity operation
\cite{HoyerSpalek2005Fanout,Moore1999FanoutParityCounting}.  The required fan-out is implemented along a strip in band \(k\) by the
measurement-based path fan-out of Lemma~\ref{lem:fanout-path}.  The parity strips for different \(k\)'s are disjoint and run
in parallel.  Their total area is also \(O(R\,\mathrm{polylog}\,R)\).

Finally, the decoder is used as a clean subroutine.  Candidate qubits,
equality-test registers, selector terms, and parity registers are computed
into clean workspace; the decoded bits are copied to the output register; and
then the decoder computation is reversed.  Thus all temporary registers are
restored to \(\ket0\), the redundant encoding \(\ket{\Phi_R(t)}\) is restored,
and the output register contains \(\ket t\).  The depth is constant, and the
total additional clean ancillary cost is \(O(R\,\mathrm{polylog}\,R)\).
\end{proof}

\begin{lemma}[Weighted counting on a 2D grid]
\label{lem:2d-dynamic-weighted-counting}
Let \(z_1,\ldots,z_M\in\{0,1\}\) be input bits with fixed nonnegative integer
weights \(a_1,\ldots,a_M\).  Suppose
$ 
    S=\sum_{i=1}^M a_i z_i
$ 
satisfies \(0\le S\le R\).  Then the weighted counting map
\[
    \ket{z}\ket{0^q}\mapsto\ket{z}\ket{S},
    \qquad \forall z=(z_1,\ldots,z_M)\in\{0,1\}^M,~q=\lceil\log(R+1)\rceil,
\]
can be implemented in constant depth in the dynamic
2D model, using
$ 
    O\bigl(MR+R\,\mathrm{polylog}\,R\bigr)
$ 
additional clean ancillary qubits.  The construction can be arranged so that
the \(M\) input bits and the \(q\)-qubit output register lie on designated
boundary ports of the patch.
\end{lemma}

\begin{proof}
If \(R=0\), the claim is trivial.  Assume \(R\ge 1\).

First apply Lemma~\ref{lem:2d-dynamic-weighted-phase-preparation}.  For every
branch \((k,y)\), this prepares one copy of the phase state
\(\ket{\phi_{k,y}(S)}\), using \(O(MR)\) clean ancillary qubits and constant
depth.  Equivalently, it prepares the redundant phase encoding required by the
Takahashi--Tani decoder.

Next apply Lemma~\ref{lem:2d-dynamic-parallel-decoder} to these redundant
phase states.  The decoder writes the binary representation of \(S\) into the
\(q\)-qubit output register in constant depth, using
\(O(R\,\mathrm{polylog}\,R)\) additional clean ancillary qubits.  By the
clean-subroutine form of the decoder, its temporary registers are restored to
\(\ket{0}\), while the redundant phase states remain available.

It remains only to unprepare the redundant phase states.  We reverse the
weighted phase-preparation circuit from
Lemma~\ref{lem:2d-dynamic-weighted-phase-preparation}.  This clears the phase
registers, copy registers, and cat-state leaves, while leaving the input bits
\(z_i\) and the output register \(\ket{S}\) unchanged.  Hence all temporary
registers are restored to \(\ket{0}\), and the desired map
\[
    \ket{z}\ket{0^q}
    \longmapsto
    \ket{z}\ket{S}
\]
is implemented exactly.

The total ancillary cost is
\[
    O(MR)+O(R\,\mathrm{polylog}\,R)
    =
    O\bigl(MR+R\,\mathrm{polylog}\,R\bigr),
\]
and the depth remains constant.  Finally, the rectangular phase-preparation
layout and the decoder layout may be placed adjacent to each other so that the
phase-state registers produced by the first part feed directly into the
decoder, and the input and output registers are accessible from the boundary
of the overall patch.
\end{proof}

As a first consequence, we obtain a direct constant-depth implementation of
Hamming weight on a two-dimensional dynamic grid.

\begin{theorem}
\label{thm:2d-dynamic-direct}
In the dynamic 2D model, $\HWC_n$ can be
implemented with constant depth and with $O(n^2)$ additional clean
ancillary qubits.
\end{theorem}

\begin{proof}
Apply Lemma~\ref{lem:2d-dynamic-weighted-counting} with $M=n$, with
$a_i=1$ for all $i$, and with $R=n$.  Then the weighted sum is
$S=\sum_i x_i=|x|$, so the lemma directly implements $\HWC_n$.  The ancillary
cost is $O(n^2+n\,\mathrm{polylog}\,n)=O(n^2)$, and the depth is constant.
\end{proof}

We next combine the two-dimensional dynamic HWC with blocking.

\begin{theorem}
\label{thm:2d-dynamic-blocked}
Let $2\le B\le n$.  In the dynamic 2D model,
$\HWC_n$ can be implemented with constant depth and with
$O(nB+n^2\log B/B+n\,\mathrm{polylog}\,n)$ additional clean ancillary qubits.
\end{theorem}

\begin{proof}
Partition the $n$ input bits into $L=\lceil n/B\rceil$ blocks, each of size at
most $B$.  For each block, apply Theorem~\ref{thm:2d-dynamic-direct} to compute
its local Hamming weight.  A block of size $B$ uses $O(B^2)$ ancillary qubits
and constant depth.  Since the blocks are disjoint, all local computations are
performed in parallel, and their total ancillary cost is $L\cdot O(B^2)=O(nB)$.

Let $w_\ell$ be the local Hamming weight of block $\ell$.  The local weight
registers contain $M=O((n/B)\log B)$ bits in total.  Writing
$w_\ell=\sum_j 2^j b_{\ell,j}$, we have
$|x|=\sum_{\ell,j}2^j b_{\ell,j}$.  Thus the bits $b_{\ell,j}$ form a weighted
counting instance with $M=O((n/B)\log B)$ input bits and range $0\le |x|\le n$.

We now describe the layout connecting the local block computations to the
global weighted-counting instance; see
Fig.~\ref{fig:2d-dynamic-blocked-layout}.  Place the $L$ local block patches in
a nearly square region.  Since their total area is $O(nB)$, this region can be
chosen to have diameter $D=O(\sqrt{nB})$.  Separately, place a rectangular
patch for Lemma~\ref{lem:2d-dynamic-weighted-counting}, with $M$ input
registers on one boundary and with range parameter $R=n$.

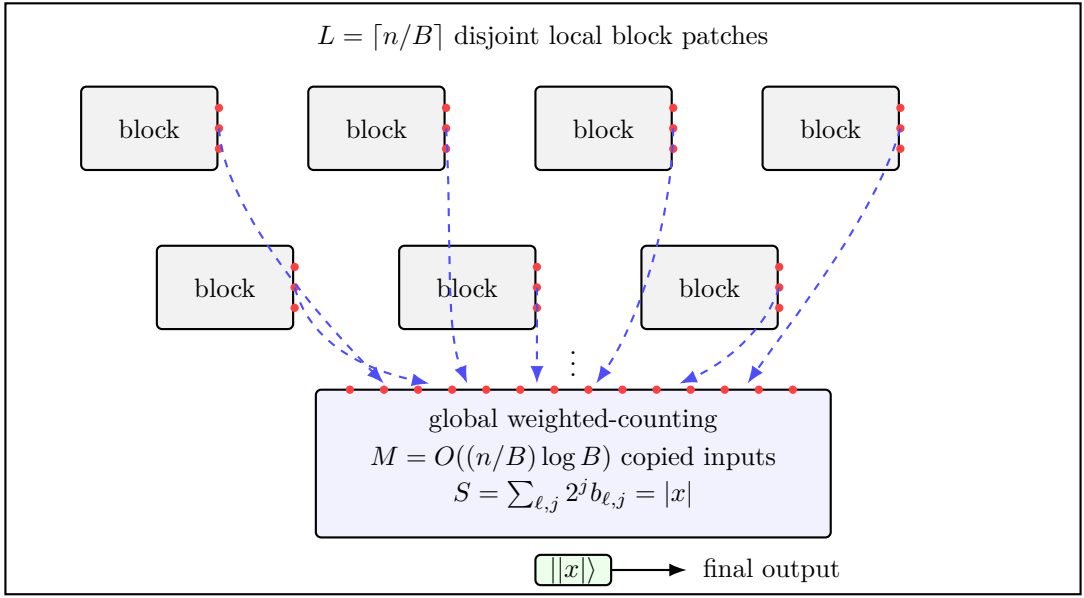
\begin{figure}[t]
\centering
\begin{tikzpicture}[
    x=1cm,y=1cm,
    every node/.style={font=\small},
    outer/.style={draw,thick},
    block/.style={draw,thick,rounded corners=2pt,fill=gray!10},
    gpatch/.style={draw,thick,rounded corners=2pt,fill=blue!5},
    port/.style={circle,fill=red!75,inner sep=1.1pt},
    pathcopy/.style={blue!70,dashed,thick,->},
    >=Latex
]

\draw[outer] (0,0) rectangle (14.2,7.8);

\node at (7.1,7.35) {$L=\lceil n/B\rceil$ disjoint local block patches};

\draw[block] (1.0,5.6) rectangle (2.8,6.7);
\draw[block] (4.0,5.6) rectangle (5.8,6.7);
\draw[block] (7.0,5.6) rectangle (8.8,6.7);
\draw[block] (10.0,5.6) rectangle (11.8,6.7);

\draw[block] (2.0,3.5) rectangle (3.8,4.6);
\draw[block] (5.2,3.5) rectangle (7.0,4.6);
\draw[block] (8.4,3.5) rectangle (10.2,4.6);

\foreach \x/\y in {1.9/6.15,4.9/6.15,7.9/6.15,10.9/6.15,2.9/4.05,6.1/4.05,9.3/4.05} {
    \node at (\x,\y) {block};
}

\foreach \p in {(2.82,6.42),(2.82,6.15),(2.82,5.88),
                (5.82,6.42),(5.82,6.15),(5.82,5.88),
                (8.82,6.42),(8.82,6.15),(8.82,5.88),
                (11.82,6.42),(11.82,6.15),(11.82,5.88),
                (3.82,4.32),(3.82,4.05),(3.82,3.78),
                (7.02,4.32),(7.02,4.05),(7.02,3.78),
                (10.22,4.32),(10.22,4.05),(10.22,3.78)} {
    \node[port] at \p {};
}

\draw[gpatch] (4.1,0.75) rectangle (10.9,2.70);

\foreach \x in {4.55,5.00,5.45,5.90,6.35,6.80,7.25,7.70,8.15,8.60,9.05,9.50,9.95,10.40} {
    \node[port] at (\x,2.70) {};
}

\node at (7.5,2.30) {global weighted-counting};
\node at (7.5,1.80) {$M=O((n/B)\log B)$ copied inputs};
\node at (7.5,1.30) {$S=\sum_{\ell,j}2^j b_{\ell,j}=|x|$};

\draw[block,fill=green!8] (7.0,0.12) rectangle (8.0,0.52);
\node at (7.5,0.32) {$\ket{|x|}$};
\draw[->,thick] (8.0,0.32) -- (9.0,0.32);
\node[anchor=west] at (9.08,0.32) {final output};

\draw[pathcopy] (2.82,6.15) .. controls (2.9,5.0) and (4.3,3.5) .. (5.00,2.78);
\draw[pathcopy] (5.82,6.15) .. controls (5.9,4.8) and (5.8,3.5) .. (6.10,2.78);
\draw[pathcopy] (8.82,6.15) .. controls (8.8,4.8) and (8.3,3.5) .. (7.80,2.78);
\draw[pathcopy] (11.82,6.15) .. controls (11.5,5.0) and (10.3,3.5) .. (9.80,2.78);

\draw[pathcopy] (3.82,4.05) .. controls (4.0,3.4) and (4.6,3.0) .. (5.60,2.78);
\draw[pathcopy] (7.02,4.05) -- (7.02,2.78);
\draw[pathcopy] (10.22,4.05) .. controls (10.0,3.4) and (9.4,3.0) .. (8.90,2.78);

\node at (7.5,3.15) {$\vdots$};

\end{tikzpicture}
\caption{Schematic layout for Theorem~\ref{thm:2d-dynamic-blocked}.  The
two-dimensional grid contains $L=\lceil n/B\rceil$ disjoint local block
patches and one global weighted-counting patch.  Each local patch computes a
block Hamming weight, and the corresponding bits $b_{\ell,j}$ are copied by
constant-depth dynamic long-range CNOTs along the indicated dashed paths to
input qubits of the global patch.  The global patch then computes
$S=\sum_{\ell,j}2^j b_{\ell,j}=|x|$ and writes the result to the final output
register.}
\label{fig:2d-dynamic-blocked-layout}
\end{figure}

For each local output bit $b_{\ell,j}$, reserve one clean input qubit on the
input boundary of the global weighted-counting rectangle.  Connect
$b_{\ell,j}$ to this input qubit by a path in the routing corridor.  The $M$
paths can be chosen vertex-disjoint by assigning one track to each copied bit.
The total number of routing qubits needed is $O(M(D+M))$: the term $O(MD)$
connects the local region to the corridor, and the term $O(M^2)$ assigns the
$M$ paths to the $M$ ordered input positions on the boundary of the global
rectangle. This routing overhead is absorbed by the stated bound.  Indeed, using
$M=O((n/B)\log B)$ and $D=O(\sqrt{nB})$, we have
\[
    M(D+M)
    =
    O\!\left(
        \frac{n^{3/2}\log B}{\sqrt B}
        +
        \frac{n^2\log^2 B}{B^2}
    \right)
    =
    O\!\left(\frac{n^2\log B}{B}\right),
\]
where we used $B\le n$ and $\log B\le B$.

Along these disjoint paths, implement the constant-depth dynamic long-range
CNOT copies using the measurement-based path fan-out of
Lemma~\ref{lem:fanout-path} together with the leaf-recovery procedure of
Lemma~\ref{lem:2d-dynamic-fanout-path}, thereby copying each local output bit
$b_{\ell,j}$ into its reserved global input qubit.  Since the paths are
disjoint, all these copies are performed in parallel and in constant depth.

Apply Lemma~\ref{lem:2d-dynamic-weighted-counting} to the copied global input
registers with $R=n$.  This writes the global Hamming weight into the final
output register in constant depth and uses
$O(Mn+n\,\mathrm{polylog}\,n)=O(n^2\log B/B+n\,\mathrm{polylog}\,n)$ additional
clean ancillary qubits.

After the global Hamming weight has been written into the final output
register, the global weighted-counting subroutine has restored its temporary
workspace to $\ket{0}$.  We then reverse the long-range CNOT copies along the
routing paths, clearing the copied global input registers.  Finally, reverse
all local block computations in parallel.  This clears the local weight
registers and local workspaces, while leaving the final output register in the
state $\ket{|x|}$.

Adding the local block workspace, the routing workspace, and the global
weighted-counting workspace gives
$O(nB+n^2\log B/B+n\,\mathrm{polylog}\,n)$ ancillary qubits.  The total depth
is constant, because it is a constant number of constant-depth stages.
\end{proof}

\begin{corollary}[Subquadratic-ancilla constant-depth HWC on a 2D grid]
\label{cor:2d-dynamic-subquadratic}
In dynamic 2D model, $\HWC_n$ can be
implemented with constant depth and with
$O(n^{3/2}\sqrt{\log n})$ additional clean ancillary qubits.
\end{corollary}

\begin{proof}
Apply Theorem~\ref{thm:2d-dynamic-blocked} with
$B=\Theta(\sqrt{n\log n})$.  This choice satisfies $2\le B\le n$ for
sufficiently large $n$.  Then
$nB=O(n^{3/2}\sqrt{\log n})$ and
$n^2\log B/B=O(n^{3/2}\sqrt{\log n})$.  The remaining
$n\,\mathrm{polylog}\,n$ term is asymptotically smaller than
$n^{3/2}\sqrt{\log n}$, so it is absorbed into the same bound.  The depth
remains constant by Theorem~\ref{thm:2d-dynamic-blocked}.
\end{proof}

We can also iterate the blocking idea.  As in the all-to-all dynamic
construction, the cleanest implementation is a pyramid: compute all levels
from bottom to top, copy the final output, and then reverse all levels.

\begin{theorem}[Recursive HWC on a 2D grid]
\label{thm:2d-dynamic-recursive}
For every integer $r\ge 1$, $\HWC_n$ can be implemented in the
dynamic 2D model with depth $O(r)$ and with
$O(r\,n^{1+1/r}\,\mathrm{polylog}\,n)$ additional clean ancillary qubits.
\end{theorem}

\begin{proof}
Let $s=n^{1/r}$, so that $B_j=s^j$ for $j=0,1,\ldots,r$.  We ignore integer
roundings in the block sizes; replacing them by ceilings changes only constant
factors.  We also fix a sufficiently large constant $c$ and write
$\Lambda=(\log n)^c$.  The power $c$ is chosen large enough to absorb all
polylogarithmic factors from Lemma~\ref{lem:2d-dynamic-weighted-counting} and
from the routing estimates below.

We use the clean XOR-output form of the weighted-counting primitive: given
weighted input bits $z$ and a target register $t$, it maps
$\ket{z}\ket{t}$ to $\ket{z}\ket{t\oplus S}$, where
$S=\sum_i a_i z_i$, while restoring all temporary registers to $\ket{0}$.
This form is obtained from the clean output version by writing the computed
bits of $S$ into the target by CNOTs and then clearing the internal work
registers.  Thus applying the same level computation twice clears its target
register.

For each $j\ge 1$, define
\[
    A_j := C\,j\,B_j s\,\Lambda,
\]
where $C>0$ is a sufficiently large constant.  We use the following slightly
stronger geometric invariant: every level-$j$ block is contained in a
rectangular region of width
\[
    W_j=O\bigl(B_j\Lambda^{1/2}\bigr)
\]
and height
\[
    H_j=O\bigl(js\Lambda^{1/2}\bigr).
\]
In particular, its area is $O(W_jH_j)=O(A_j)$.

Consider one level-$j$ block.  Its $s$ children together provide
\[
    M_j
    =
    s\left\lceil\log_2(B_{j-1}+1)\right\rceil
    =
    O(s\log n).
\]
input bits to the weighted-counting primitive.  By Lemma~\ref{lem:2d-dynamic-weighted-counting},
the new weighted-counting patch has area
\[
    O(B_js\,\Lambda).
\]
Moreover, after increasing the constant $c$ in
$\Lambda=(\log n)^c$ if necessary, the explicit phase-preparation array and
banded decoder layout used in that lemma fit in a rectangle of width
\[
    O\bigl(B_j\Lambda^{1/2}\bigr)
\]
and height
\[
    O\bigl(s\Lambda^{1/2}\bigr).
\]
We orient this rectangle so that its input ports lie on the side facing the
child regions.

For $j\ge 2$, place the $s$ level-$(j-1)$ child regions side by side.  By the
induction hypothesis, their combined width is
\[
    sW_{j-1}
    =
    O\bigl(sB_{j-1}\Lambda^{1/2}\bigr)
    =
    O\bigl(B_j\Lambda^{1/2}\bigr),
\]
while their height is
\[
    O\bigl((j-1)s\Lambda^{1/2}\bigr).
\]
For $j=1$, the $s$ input sites clearly satisfy the same bound.  Place the new
weighted-counting patch in an adjacent horizontal strip.  The resulting
pre-routing rectangle has width
\[
    O\bigl(B_j\Lambda^{1/2}\bigr)
\]
and height
\[
    O\bigl(js\Lambda^{1/2}\bigr),
\]
and hence diameter
\[
    D_j
    =
    O\bigl(B_j\Lambda^{1/2}+js\Lambda^{1/2}\bigr).
\]

For each of the $M_j$ child output bits, reserve one clean input qubit on the
input boundary of the weighted-counting patch.  Route the child outputs to
these reserved inputs through a corridor between the child regions and the
patch.  Using one track per copied bit, the paths can be chosen
vertex-disjoint with total routing area
\[
\begin{aligned}
    O\bigl(M_j(D_j+M_j)\bigr)
    &=
    O\Bigl(
        s\log n\,
        \bigl(
            B_j\Lambda^{1/2}
            +js\Lambda^{1/2}
            +s\log n
        \bigr)
      \Bigr)                                                     \\
    &=
    O\bigl(B_js\,\Lambda\bigr).
\end{aligned}
\]
Here we used $B_j\ge s$, $j\le r$, and the fact that $r$ is independent of
$n$; all remaining logarithmic factors are absorbed by choosing $c$
sufficiently large.  The routing corridor can therefore be accommodated by
adding height $O(s\Lambda^{1/2})$ to the rectangle, which preserves the bounds
$W_j=O(B_j\Lambda^{1/2})$ and $H_j=O(js\Lambda^{1/2})$.

Consequently, the child regions, the weighted-counting patch, and the routing
corridor together occupy area
\[
    O\bigl((j-1)B_js\,\Lambda\bigr)
    +O\bigl(B_js\,\Lambda\bigr)
    =
    O\bigl(jB_js\,\Lambda\bigr)
    =
    O(A_j).
\]

Now run the computation level by level.  At level $j$, all level-$j$ blocks
are disjoint, so they are processed in parallel.  Within each block, first use
the constant-depth dynamic long-range CNOT primitive along the disjoint routing
paths to copy the child output bits into the input registers of the
weighted-counting patch.  Then run the weighted-counting patch to XOR the
level-$j$ Hamming weight into its level-$j$ output register.  Finally undo the
long-range CNOT copies, clearing the copied input registers.  This whole
level has constant depth and leaves only the level-$j$ weight registers in
addition to the lower-level registers.

Performing the forward pass for $j=1,2,\ldots,r$ therefore takes depth $O(r)$.
The top level contains the global Hamming weight.  Copy this value to the
final output register using the same constant-depth dynamic long-range CNOT
primitive; the extra $O(\log n)$ auxiliary qubits are absorbed into the bound.

We then clear the intermediate level registers from top to bottom.  For
$j=r,r-1,\ldots,1$, apply the same clean XOR-output level-$j$ computation
again, using the level-$(j-1)$ registers as inputs and the stored level-$j$
weight registers as targets.  Since each target currently contains exactly the
corresponding Hamming weight, the XOR action maps it back to $\ket{0}$.  The
temporary copied inputs, routing workspaces, and weighted-counting workspaces
are cleared within the same constant-depth level computation.  The copied
final output register is not touched.  The clearing pass also has depth
$O(r)$.

It remains to count ancillary qubits.  The top-level recursive layout has area
\[
    A_r
    =
    O(rB_rs\,\Lambda)
    =
    O(rn^{1+1/r}\,\mathrm{polylog}\,n).
\]
This area includes the regions for all lower-level blocks, all routing
corridors, all weighted-counting workspaces, and all intermediate weight
registers.  Excluding the original input register and the final output
register can only decrease the count.  Hence the number of additional clean
ancillary qubits is
$O(rn^{1+1/r}\,\mathrm{polylog}\,n)$, and the total depth is $O(r)$.
\end{proof}

\begin{corollary}
\label{cor:2d-dynamic-near-linear}
For every fixed $\varepsilon>0$, $\HWC_n$ can be implemented in the
dynamic 2D model with depth $O_\varepsilon(1)$
and with $O(n^{1+\varepsilon}\,\mathrm{polylog}\,n)$ additional clean ancillary
qubits.
\end{corollary}

\begin{proof}
Choose $r=\max\{1,\lceil 1/\varepsilon\rceil\}$.  Then $r$ is a constant
depending only on $\varepsilon$, and $1/r\le \varepsilon$.  By
Theorem~\ref{thm:2d-dynamic-recursive}, the depth is
$O(r)=O_\varepsilon(1)$, and the ancillary cost is
$O(r\,n^{1+1/r}\,\mathrm{polylog}\,n)$.  Since $n^{1+1/r}\le n^{1+\varepsilon}$
and $r$ is absorbed into the constant for fixed $\varepsilon$, the ancillary
cost is $O(n^{1+\varepsilon}\,\mathrm{polylog}\,n)$.
\end{proof}

\section{Circuit Depth and Ancilla Count for Symmetric Boolean Functions}

\label{sec:symmetric-functions}

\begin{definition}[Symmetric Boolean function]
\label{def:symmetric-boolean}
A Boolean function \(f:\{0,1\}^n\to\{0,1\}\) is called
\emph{symmetric} if
\[
    f(x_1,\ldots,x_n)
    =
    f(x_{\pi(1)},\ldots,x_{\pi(n)})
\]
for every \(x=(x_1,\ldots,x_n)\in\{0,1\}^n\) and every permutation
\(\pi\in S_n\).  Equivalently, there exists a function
\(g:\{0,1,\ldots,n\}\to\{0,1\}\) such that \(f(x)=g(|x|)\).
\end{definition}

The natural implementation is to compute the Hamming weight into a temporary register,
evaluate \(g\) on that register, and then uncompute the Hamming weight.

We first record the Boolean-function primitives used for the post-processing
step.

\begin{lemma}
\label{lem:boolean-standard-alltoall}
Let \(g:\{0,1\}^k\to\{0,1\}\) be an arbitrary Boolean function.  In the
standard all-to-all model, the map
\[
    \ket{y}\ket{b}
    \longmapsto
        \ket{y}\ket{b\oplus g(y)},\quad \forall y\in\{0,1\}^k,~\forall b\in\{0,1\},
\]
can be implemented with depth \(O(k)\), using \(O(2^k/k)\) clean ancillary
qubits and \(O(2^k)\) borrowed qubits.
\end{lemma}

\begin{proof}
Let \(N=2^k\), and assume \(k\) is sufficiently large.  Set
\(\ell=\lceil\log k\rceil\), write \(y=(u,v)\) with
\(u\in\{0,1\}^r\), \(v\in\{0,1\}^{\ell}\), and \(r=k-\ell\).  Let
\(P=2^r\).  Then \(P=O(2^k/k)\) and \(2^\ell=O(k)\).  We view \(u\) as a
block address and \(v\) as an address inside a block.  For each
\(a\in\{0,1\}^r\), define \(h_a(v)=g(a,v)\).

We first compute the high-address one-hot flags \(E_a=[u=a]\) for all
\(a\in\{0,1\}^r\).  This is done by a reversible binary decoder.  For clarity,
we describe one layer.  Suppose that after \(j\) layers we have clean prefix
flags \(F_p=[u_1\cdots u_j=p]\) for all \(p\in\{0,1\}^j\).  To extend the
prefix by the next bit \(u_{j+1}\), we create two child flags
\(F_{p0}\) and \(F_{p1}\), initialized to \(\ket0\), by applying
\[
    F_{p1}\leftarrow F_{p1}\oplus F_p u_{j+1},
    \qquad
    F_{p0}\leftarrow F_{p0}\oplus F_p(1-u_{j+1}).
\]
The second operation is implemented by conjugating the \(u_{j+1}\) control by
\(X\) gates.  Since the same input bit \(u_{j+1}\) is needed for all
\(2^j\) prefixes at level \(j\), we first make \(2^j\) clean CNOT copies of
\(u_{j+1}\), one for each prefix.  Then all prefix-splitting operations at
this level act on disjoint registers and can be performed in constant depth.
The copies of all bits of \(u\) can be prepared by CNOT trees using
\(O(P)\) clean qubits and depth \(O(\log P)=O(k)\).  The decoder has
\(r=O(k)\) levels and uses a total of \(O(P)\) clean prefix flags.  Therefore
the high one-hot encoding has depth \(O(k)\) and clean workspace \(O(P)\).

Similarly, compute the low-address one-hot flags \(Q_\beta=[v=\beta]\) for all
\(\beta\in\{0,1\}^{\ell}\) by the same decoder construction.  Since
\(2^\ell=O(k)\), this costs depth \(O(k)\) and \(O(k)\) clean qubits.

Next we compute, for every high block \(a\), the local value \(Z_a=h_a(v)\)
into a clean qubit \(Z_a\).  Since the \(Q_\beta\)'s form a one-hot encoding
of \(v\), we have \(Z_a=\bigoplus_{\beta:\,h_a(\beta)=1} Q_\beta\).  Thus
\(O(P)\) clean qubits suffice for the registers \(Z_a\).  For each selected
pair \((a,\beta)\) with \(h_a(\beta)=1\), introduce one borrowed qubit
\(d_{a,\beta}\).  There are at most \(P2^\ell=2^k\) such borrowed qubits.

We now explain how a selected pair \((a,\beta)\) contributes \(Q_\beta\) to
\(Z_a\) without using a clean copy of \(Q_\beta\).  First apply the 
\(\CNOT^{d_{a,\beta}}_{Z_a}\).  Then perform dirty fan-out from \(Q_\beta\) to all
borrowed qubits \(d_{a,\beta}\) with this fixed value of \(\beta\), so that
\(d_{a,\beta}\mapsto d_{a,\beta}\oplus Q_\beta\).  By
Lemma~\ref{lem:dirty-fanout-tree}, this dirty fan-out can be implemented on a
balanced binary tree of height \(O(\log P)\), and hence has depth
\(O(\log P)=O(k)\) in the all-to-all model.  The fan-outs for different values
of \(\beta\) act on disjoint borrowed target sets and can be done in parallel.
After the fan-out, apply the same CNOTs \(d_{a,\beta}\to Z_a\) again, and then
undo the dirty fan-outs.

For each selected pair, the two CNOTs into \(Z_a\) contribute
\[
    d_{a,\beta}\oplus(d_{a,\beta}\oplus Q_\beta)=Q_\beta .
\]
Thus the unknown borrowed value cancels, and the borrowed qubit is restored
after the inverse dirty fan-out.  The CNOT gates \(\CNOT^{d_{a,\beta}}_{Z_a}\) are
scheduled in \(O(2^\ell)=O(k)\) layers according to the value of \(\beta\),
since for a fixed \(\beta\) the targets \(Z_a\) are distinct.  Therefore all
local values \(Z_a=h_a(v)\) are computed in depth \(O(k)\).

Finally, compute \(T_a=E_a Z_a\) into clean qubits \(T_a\) for all
\(a\in\{0,1\}^r\).  These Toffoli gates are disjoint and have constant depth.
Since exactly one high-address flag \(E_a\) equals \(1\), we have
\[
    \bigoplus_a T_a = g(u,v)=g(y).
\]
XOR this parity into the target qubit \(b\) using a CNOT tree of depth
\(O(\log P)=O(k)\), and then reverse the tree so that all \(T_a\) registers
are restored.  Then uncompute the \(T_a\) registers, the \(Z_a\) registers,
the low flags \(Q_\beta\), the high flags \(E_a\), and all temporary input
copies.

All clean ancillary qubits return to \(\ket0\), and all borrowed qubits return
to their initial states.  The clean workspace is
\(O(P)+O(2^\ell)+O(P)+O(P)=O(P)=O(2^k/k)\), and the borrowed workspace is at
most \(P2^\ell=2^k\).  Each stage has depth \(O(k)\), and the uncomputation
only changes the depth by a constant factor.  This proves the claim.
\end{proof}

\begin{lemma}
\label{lem:boolean-dynamic-alltoall}
Let \(g:\{0,1\}^k\to\{0,1\}\) be an arbitrary Boolean function.  In the dynamic
all-to-all model, the map
\[
    \ket{y}\ket{b}
    \longmapsto
    \ket{y}\ket{b\oplus g(y)},\quad \forall y\in\{0,1\}^k,~\forall b\in\{0,1\},
\]
can be implemented with constant depth and with
\(O(2^k\,\mathrm{polylog}\,2^k)\) clean ancillary qubits.
\end{lemma}

\begin{proof}
We use exact constant-depth OR and AND primitives in the unbounded-fan-out
model~\cite{TakahashiTani2016Collapse}.  Each fan-out gate is implemented in
constant depth by measurement-based fan-out in the dynamic circuit model.

For each \(t\in\{0,1\}^k\), we compute an equality flag
\(e_t=[y=t]\).  To make all these computations parallel, first fan out each
input bit \(y_i\) to one private copy for each string \(t\in\{0,1\}^k\).
This uses \(O(k2^k)\) clean ancillary qubits and constant depth.  For the
branch corresponding to \(t\), we use the private copies of the bits \(y_i\).
If \(t_i=0\), we flip the corresponding copy before and after the AND
computation; equivalently, we use the literal \(1-y_i\).  Thus \(e_t\) is the
AND of \(k\) literals.

Using the constant-depth AND primitive, all \(2^k\) equality flags can now be
computed in parallel, since different branches use disjoint registers.  Each
branch uses \(\mathrm{poly}(k)\) clean ancillary qubits, so the total ancillary
cost for the one-hot encoding is
$ 
    O(2^k\mathrm{poly}\,k).
$

Let \(G=\{t\in\{0,1\}^k:g(t)=1\}\).  For every basis input \(y\), exactly one
equality flag \(e_t\) is equal to \(1\).  Therefore
\[
    g(y)=\bigvee_{t\in G} e_t .
\]
We compute this OR into one clean qubit using the constant-depth OR primitive,
CNOT this qubit into the target qubit \(b\), and then reverse the OR
computation.  Finally, we uncompute all equality flags and undo the initial
fan-out copies of the input bits.  This restores all ancillary qubits to
\(\ket0\) and implements
\[
    \ket{y}\ket{b}
    \longmapsto
    \ket{y}\ket{b\oplus g(y)}, \forall y\in\{0,1\}^k,~\forall b\in\{0,1\}
\]

The depth is constant.  The ancillary cost is
\[
    O(k2^k)+O(2^k\mathrm{poly}\,k)+O(2^k\mathrm{polylog}\,2^k)
    =
    O(2^k\,\mathrm{polylog}\,2^k),
\]
because \(k=\log 2^k\).
\end{proof}

The next two lemmas are used only for the measurement-free two-dimensional
model.  They allow us to evaluate an arbitrary Boolean function on \(k\)
input bits in depth \(O(\sqrt{2^k})\), using borrowed rather than clean
workspace.

We use the standard mesh-routing fact that any permutation of \(S\) qubits on
an \(O(\sqrt S)\times O(\sqrt S)\) square-grid patch can be implemented by
nearest-neighbor SWAP gates with depth \(O(\sqrt S)\)
\cite{LeightonMakedonTollis1995Routing}.  The constant-size local buffers in
mesh-routing algorithms are absorbed into the \(O(S)\) borrowed workspace used
below, and no clean ancillary qubits are required.

We also use the \(\mathsf H\)-dual of the dirty fan-out circuit from
Corollary~\ref{cor:dirty-fanout-2d}.  By conjugating that circuit with
Hadamard gates on all participating qubits, one obtains a parity fan-in circuit
with the same depth and the same workspace; equivalently, this is the standard
duality between fan-out and parity~\cite{HoyerSpalek2005Fanout}.  More
precisely, suppose that the qubits \(r_1,\ldots,r_R\) and a target qubit \(b\)
are placed on the vertices of a connected tree of height \(h\), rooted at
\(b\), and that the vertex set of the tree consists exactly of these qubits.
Then the map
\[
    \ket{r_1,\ldots,r_R}\ket{b}
    \longmapsto
    \ket{r_1,\ldots,r_R}
    \ket{b\oplus r_1\oplus\cdots\oplus r_R}
\]
can be implemented with depth \(O(h)\), using no clean ancillary qubits and
restoring all qubits \(r_i\).  In particular, after two-dimensional SWAP
routing, such a parity fan-in on \(R\) selected qubits in a square-grid patch
has depth \(O(\sqrt R)\).

\begin{lemma}[Shared-control Toffoli batches on a two-dimensional grid]
\label{lem:2d-shared-control-toffoli}
Consider \(S\) Toffoli gates of the form
\[
    b_i \longmapsto b_i\oplus c a_i,
    \qquad i=1,\ldots,S,
\]
where the qubits \(a_i,b_i\) are pairwise distinct and are all distinct from
the common control qubit \(c\).  Suppose these qubits, together with
\(O(S)\) borrowed qubits, are placed in a connected two-dimensional
square-grid patch of area \(O(S)\).  Then the whole batch can be implemented
with depth \(O(\sqrt S)\), using \(O(S)\) borrowed qubits and no clean
ancillary qubits.
\end{lemma}

\begin{proof}
First use two-dimensional SWAP routing to arrange the qubits into \(S\)
constant-size cells inside the patch, so that the \(i\)-th cell contains
\(a_i\), \(b_i\), and a borrowed qubit \(d_i\).  Place the common control
\(c\) at a fixed root position of the patch.  This routing has depth
\(O(\sqrt S)\), and it will be reversed at the end.

We use the borrowed qubits \(d_i\) to remove the shared-control conflict.
First apply
\[
    \Tof(d_i,a_i;b_i),
    \qquad i=1,\ldots,S,
\]
in parallel.  Each Toffoli gate acts inside a constant-size cell and is
implemented by a constant-size nearest-neighbor circuit over one- and
two-qubit gates.  Thus this layer contributes only constant depth.

Next perform a dirty fan-out from the root \(c\) to the borrowed qubits
\(d_i\).  Choose a bounded-degree tree inside the patch, rooted at \(c\), whose
targets include all \(d_i\)'s.  The other vertices of the tree, if any, are
borrowed routing qubits.  The data qubits \(a_i\) and \(b_i\) are not used as
fan-out targets.  By Corollary~\ref{cor:dirty-fanout-2d}, this dirty fan-out
has depth \(O(\sqrt S)\), uses no clean ancillary qubits, and maps each
\(d_i\) to \(d_i\oplus c\).

Now apply the same local Toffoli layer
\[
    \Tof (d_i,a_i;b_i),
    \qquad i=1,\ldots,S,
\]
again, and finally undo the dirty fan-out.  For each \(i\), the two Toffoli
layers change \(b_i\) by
\[
    a_i d_i \oplus a_i(d_i\oplus c)
    =
    a_i c .
\]
Thus the desired update \(b_i\mapsto b_i\oplus c a_i\) is implemented.  Since
the dirty fan-out is undone, every borrowed qubit used in the fan-out tree is
restored to its initial state.  Finally, reverse the initial SWAP routing.  The
total depth is \(O(\sqrt S)\), and no clean ancillary qubits are used.
\end{proof}

\begin{lemma}[Boolean evaluation with borrowed qubits on a two-dimensional grid]
\label{lem:boolean-standard-2d-borrowed}
Let \(g:\{0,1\}^k\to\{0,1\}\) be an arbitrary Boolean function.  In the
standard 2D model, the map
\[
    \ket{y}\ket{b}
    \longmapsto
    \ket{y}\ket{b\oplus g(y)},\quad \forall y\in\{0,1\}^k,~\forall b\in\{0,1\}
\]
can be implemented with depth \(O(\sqrt{2^k})\), using \(O(2^k)\) borrowed
qubits and no asymptotic clean ancillary workspace.
\end{lemma}

\begin{proof}
Let \(N=2^k\).  Write \(g\) in algebraic normal form,
\[
    g(y)
    =
    \alpha_\emptyset
    \oplus
    \bigoplus_{i=1}^k \alpha_i y_i
    \oplus
    \bigoplus_{\substack{A\subseteq[k]\\ |A|\ge2}}
        \alpha_A \prod_{i\in A} y_i .
\]
If \(\alpha_\emptyset=1\), apply an \(X\) gate to \(b\).  The linear part is
implemented by applying the parity fan-in described above to the selected input
qubits \(y_i\) with \(\alpha_i=1\) and the target \(b\).  This has depth
\(O(\sqrt N)\), since at most \(k\le N\) input qubits participate.

It remains to implement the higher-degree part.  We follow the
borrowed-ancilla ESOP construction of
Ref.~\cite{ZiNieSun2025SymmetricFunctions}.  For every subset
\(A\subseteq[k]\) with \(|A|\ge2\), introduce one borrowed register \(r_A\).
Let
\[
    M_A(y)=\prod_{i\in A} y_i .
\]
We construct a monomial-generation circuit \(G_k\) satisfying
\[
    G_k:\quad r_A\longmapsto r_A\oplus M_A(y),
    \qquad
    A\subseteq[k],\ |A|\ge2,
\]
while leaving the input bits \(y_i\) unchanged.  The registers \(r_A\) are
borrowed and are not assumed to be initialized to \(\ket{0}\).

The circuit \(G_k\) is built recursively.  For \(j\ge2\), let \(G_j\) denote
the corresponding circuit for the variables \(y_1,\ldots,y_j\).  The base case
\(G_2\) is the single Toffoli gate
\[
    r_{\{1,2\}}\longmapsto r_{\{1,2\}}\oplus y_1y_2 .
\]
Assume that \(G_j\) has been constructed.  To obtain \(G_{j+1}\), first create
the new quadratic monomials containing \(y_{j+1}\): for each \(i\in[j]\),
apply
\[
    r_{\{i,j+1\}}
    \longmapsto
    r_{\{i,j+1\}}\oplus y_i y_{j+1}.
\]
These \(j\) Toffoli gates share the common control \(y_{j+1}\).

Next consider the new monomials \(A\cup\{j+1\}\), where \(A\subseteq[j]\) and
\(|A|\ge2\).  For all such \(A\), apply the Toffoli gates
\[
    \Tof(r_A,y_{j+1};r_{A\cup\{j+1\}})
\]
in parallel as a shared-control batch, then apply \(G_j\), and then apply the
same shared-control batch again.  If the initial value of \(r_A\) is \(d_A\),
then after \(G_j\) it is \(d_A\oplus M_A(y)\).  Therefore the two Toffoli
batches change \(r_{A\cup\{j+1\}}\) by
\[
    d_A y_{j+1}
    \oplus
    (d_A\oplus M_A(y))y_{j+1}
    =
    M_A(y)y_{j+1}
    =
    M_{A\cup\{j+1\}}(y).
\]
Thus the unknown borrowed value \(d_A\) cancels, and the desired new monomial
is generated.  At the same time, \(G_j\) leaves every old register \(r_A\) in
the state \(r_A\oplus M_A(y)\).  Hence \(G_{j+1}\) has the required action.

We now bound the two-dimensional depth of \(G_k\).  Lay out all input qubits,
borrowed monomial registers, and an additional pool of \(O(N)\) borrowed
routing qubits inside a square-grid patch of area \(O(N)\).  This pool is
reused for different batches and different recursive levels.  The layout is
chosen in nested subpatches \(P_j\), where \(P_j\) has area \(O(2^j)\) and
contains all registers touched by \(G_j\).  At the step from \(G_j\) to
\(G_{j+1}\), the quadratic batch has \(j\) gates, and each of the two
higher-degree batches has \(O(2^j)\) gates.  In each batch the gates share the
common control \(y_{j+1}\), and all other controls and targets in that batch
are pairwise distinct.  By Lemma~\ref{lem:2d-shared-control-toffoli}, each
batch can be implemented on \(P_{j+1}\) with depth \(O(\sqrt{2^j})\), using
only borrowed qubits.  The routing inside each batch is reversed after the
batch, so the recursive layout is restored before the next part of the
construction is applied.  Therefore, if \(D_j\) is the depth of \(G_j\), then
\[
    D_{j+1}
    \le
    D_j + O(\sqrt{2^j}).
\]
It follows that
\[
    D_k
    =
    \sum_{j=1}^{k-1} O(\sqrt{2^j})
    =
    O(\sqrt{2^k})
    =
    O(\sqrt N).
\]

After applying \(G_k\), each higher-degree monomial register contains
\(r_A\oplus M_A(y)\).  Let
\[
    \mathcal H=\{A\subseteq[k]: |A|\ge2,\ \alpha_A=1\}.
\]
If \(\mathcal H\) is nonempty, use parity fan-in to add the parity of the
registers \(r_A\), \(A\in\mathcal H\), into \(b\).  Before applying this
parity fan-in, use two-dimensional SWAP routing to place exactly these
selected registers and the target \(b\) on a connected tree of area
\(O(|\mathcal H|+1)\), and reverse the routing afterwards.  Then apply
\(G_k^\dagger\), which restores every monomial register \(r_A\) to its initial
borrowed value.  If \(\mathcal H\) is nonempty, apply the same parity fan-in
from the registers \(r_A\), \(A\in\mathcal H\), into \(b\) once more, again
with the routing reversed afterwards.  The two parity fan-ins change \(b\) by
\[
    \bigoplus_{A\in\mathcal H} (r_A\oplus M_A(y))
    \oplus
    \bigoplus_{A\in\mathcal H} r_A
    =
    \bigoplus_{A\in\mathcal H} M_A(y).
\]
Thus all unknown borrowed values cancel, while exactly the higher-degree part
of \(g(y)\) is added to \(b\).  The circuit \(G_k^\dagger\) restores all
borrowed monomial registers, and the parity fan-in circuits restore their
input registers.

The monomial-generation circuit and its inverse have depth \(O(\sqrt N)\).
Each parity fan-in involves at most \(O(N)\) registers and has depth
\(O(\sqrt N)\).  Therefore the total depth is \(O(\sqrt N)=O(\sqrt{2^k})\).
The number of borrowed qubits is \(O(N)=O(2^k)\), and no asymptotic clean
ancillary workspace is used.
\end{proof}

\begin{lemma}[Boolean evaluation in the dynamic two-dimensional model]
\label{lem:boolean-dynamic-2d}
Let \(g:\{0,1\}^k\to\{0,1\}\) be an arbitrary Boolean function.  In the dynamic
2D model, the map
\[
    \ket{y}\ket{b}
    \longmapsto
    \ket{y}\ket{b\oplus g(y)}, \forall y\in\{0,1\}^k,~\forall b\in\{0,1\},
\]
can be implemented with constant depth and with
$ 
    O\!\left(2^k\,\mathrm{polylog}\,2^k\right)
$
clean ancillary qubits.
\end{lemma}

\begin{proof}
Let \(N=2^k\).  We use a one-hot evaluation of \(g\).  For each
\(t\in\{0,1\}^k\), we compute an equality flag
\[
    e_t=[y=t].
\]
For every fixed input \(y\), exactly one flag \(e_t\) is equal to \(1\).
Hence, if
\[
    \mathcal G=\{t\in\{0,1\}^k:g(t)=1\},
\]
then
\[
    g(y)=\bigoplus_{t\in\mathcal G} e_t .
\]

We now describe the dynamic 2D implementation.  For each
\(t\in\{0,1\}^k\), allocate a branch patch \(Q_t\).  The patch \(Q_t\)
contains private copies \(y_1^{(t)},\ldots,y_k^{(t)}\), one clean flag qubit
\(e_t\), and \(\mathrm{poly}(k)\) clean workspace qubits for a constant-depth
AND computation.  The patches \(Q_t\) are pairwise disjoint.

First, for each input bit \(y_i\), use measurement-based fan-out to copy
\(y_i\) to the clean targets \(y_i^{(t)}\), for all
\(t\in\{0,1\}^k\).  The \(k\) fan-out structures are laid out on disjoint
strips of the two-dimensional grid and are performed in parallel.  By
Lemma~\ref{lem:fanout-path}, each such fan-out has constant depth and uses
\(O(N)\) clean ancillary qubits.  Thus this step has constant depth and uses
\(O(kN)\) clean ancillary qubits.

Next, each branch \(Q_t\) computes \(e_t=[y=t]\) from its private copies.  If
the \(i\)-th bit of \(t\) is \(0\), apply an \(X\) gate to
\(y_i^{(t)}\); otherwise do nothing.  After these local \(X\) gates, the flag
\(e_t\) is the AND of the \(k\) resulting literals.  We use the exact
constant-depth AND primitive in the unbounded-fan-out model
\cite{TakahashiTani2016Collapse}, with each fan-out gate implemented by
measurement-based fan-out in the dynamic 2D model.  This requires
\(\mathrm{poly}(k)\) clean ancillary qubits inside \(Q_t\).  Since the patches
\(Q_t\) are disjoint, all equality flags are computed in parallel.  Finally,
undo the local \(X\) gates.  The total clean workspace for all branches is
\(N\,\mathrm{poly}(k)\), and the depth is constant.

It remains to add the value of \(g(y)\) to the target qubit \(b\).  Since the
flags \(e_t\) form a one-hot encoding, it suffices to add the parity of the
selected flags \(e_t\), \(t\in\mathcal G\), to \(b\).  This parity fan-in is
the \(\mathsf H\)-dual of measurement-based fan-out, so it can also be
implemented in constant depth in the dynamic 2D model, using \(O(N)\) clean
ancillary qubits.  If \(\mathcal G=\emptyset\), this step is omitted.  This
updates the target as
\[
    b\longmapsto b\oplus\bigoplus_{t\in\mathcal G} e_t
    =
    b\oplus g(y).
\]

The parity fan-in does not change the flags \(e_t\).  Therefore we can
uncompute all branch computations, restoring the flags \(e_t\) and the local
workspaces in the patches \(Q_t\) to \(\ket{0}\).  Finally, we undo the initial
fan-out of the input bits \(y_i\), which restores all private copies
\(y_i^{(t)}\) to \(\ket{0}\).  The input register \(\ket y\) is left
unchanged, and the target register is \(\ket{b\oplus g(y)}\).

The circuit consists of a constant number of constant-depth dynamic 2D layers.
The number of clean ancillary qubits is
$
    O(kN)+N\,\mathrm{poly}(k)+O(N)
    =
    O(N\,\mathrm{poly}(k)).
$
Since \(N=2^k\) and \(k=\log N\), this is
$
    O\!\left(2^k\,\mathrm{polylog}\,2^k\right).
$
\end{proof}

We can now combine these Boolean-evaluation primitives with the Hamming-weight
constructions developed earlier.

\begin{theorem}[Symmetric Boolean functions]
\label{thm:symmetric-functions}
Let \(f:\{0,1\}^n\to\{0,1\}\) be an arbitrary symmetric Boolean function.
Then the oracle
\[
    U_{\sym}^{(n)}:\ket{x}\ket{b}
    \longmapsto
    \ket{x}\ket{b\oplus f(x)},\quad \forall x\in\{0,1\}^n,~\forall b\in\{0,1\},
\]
can be implemented with the following resources.

\begin{enumerate}
    \item In the standard all-to-all model, \(U_{\sym}^{(n)}\) has depth
    \(O(\log n)\) and uses \(O(n/\log n)\) clean ancillary qubits.

    \item In the dynamic all-to-all model, for every fixed
    \(\varepsilon>0\), \(U_{\sym}^{(n)}\) has depth \(O_\varepsilon(1)\) and
    uses
    $
        O\!\left(n^{1+\varepsilon}\,\mathrm{polylog}\,n\right)
    $
    clean ancillary qubits.

    \item In the standard 2D model, \(U_{\sym}^{(n)}\) has depth
    \(O(\sqrt n)\) and uses \(O(\log^2 n)\) clean ancillary qubits.

    \item In the dynamic 2D model, for every fixed \(\varepsilon>0\),
    \(U_{\sym}^{(n)}\) has depth \(O_\varepsilon(1)\) and uses
    $
        O\!\left(n^{1+\varepsilon}\,\mathrm{polylog}\,n\right)
    $
    clean ancillary qubits.
\end{enumerate}
\end{theorem}

\begin{proof}
Let
$
    m=\lceil\log(n+1)\rceil .
$
Since \(f\) is symmetric, there exists a function
\(g:\{0,1,\ldots,n\}\to\{0,1\}\) such that \(f(x)=g(|x|)\).  We extend
\(g\) arbitrarily to a Boolean function on all \(m\)-bit strings.  This
extension does not affect the computation, because the Hamming-weight register
only takes values in \(\{0,1,\ldots,n\}\).

The general circuit is the same in all four models.  First introduce an
\(m\)-qubit clean register and compute the Hamming weight,
\[
    \ket{x}\ket{b}\ket{0^m}
    \longmapsto
    \ket{x}\ket{b}\ket{|x|},\quad x\in\{0,1\}^n,~\forall b\in\{0,1\}.
\]
Then apply the appropriate Boolean-evaluation circuit for the \(m\)-bit
function \(g\) to the Hamming-weight register and the target qubit \(b\).  This
maps
\[
    \ket{x}\ket{b}\ket{|x|}
    \longmapsto
    \ket{x}\ket{b\oplus g(|x|)}\ket{|x|}
    =
    \ket{x}\ket{b\oplus f(x)}\ket{|x|},\quad x\in\{0,1\}^n,~\forall b\in\{0,1\}.
\]
Finally, apply the inverse Hamming-weight computation to restore the temporary
Hamming-weight register to \(\ket{0^m}\).  In the dynamic models, this means
the logical inverse of the Hamming-weight construction; the measurement-based
fan-out and long-range CNOT primitives implement deterministic logical
operations with the same asymptotic resources.  The temporary \(m\)-qubit
Hamming-weight register is counted as clean ancillary workspace, but
\(m=O(\log n)\) is absorbed by all bounds below.

We will use the following observation for the standard-model Boolean
post-processing.  The Boolean-evaluation primitives in
Lemmas~\ref{lem:boolean-standard-alltoall} and
\ref{lem:boolean-standard-2d-borrowed} use borrowed qubits.  In the present
composition, these borrowed qubits are supplied by the original input register
\(\ket{x}\), which is not modified during the Boolean post-processing and is
restored exactly by the borrowed-qubit primitives.  To avoid constant-factor
issues in the inequality \(2^m=O(n)\), fix a sufficiently large constant
\(c\).  Write the Hamming-weight register as \((u,v)\), where \(u\) consists
of the \(c\) most significant bits and \(v\) consists of the remaining
\(m-c\) bits.  For each \(s\in\{0,1\}^c\), compute the constant-size flag
\(q_s=[u=s]\), apply the Boolean-evaluation primitive to the function
\[
    (q_s,v)\longmapsto q_s\,g(s,v)
\]
on \(m-c+1\) input bits, and then uncompute \(q_s\).  Since \(2^c=O(1)\), this
only changes the depth and clean workspace by constant factors.  By choosing
\(c\) large enough, the number of borrowed qubits required in each such
Boolean-evaluation call is at most \(n\), so the original input register
provides enough borrowed workspace.

We now instantiate this scheme in each model.

In the standard all-to-all model, use
Corollary~\ref{cor:alltoall-sublinear} for the Hamming-weight computation and
Lemma~\ref{lem:boolean-standard-alltoall} for the Boolean post-processing,
with the constant-bit splitting described above.  The Hamming-weight
computation has depth \(O(\log n)\) and uses
$
    O\!\left(\frac{n\sqrt{\log n}}{2^{\sqrt{\log n}}}\right)
$
clean ancillary qubits.  Each Boolean-evaluation call has input length
\(m-c+1=\Theta(\log n)\), depth \(O(\log n)\), and clean ancillary cost
$
    O\!\left(\frac{2^{m-c+1}}{m-c+1}\right)
    =
    O(n/\log n).
$
There are only \(2^c=O(1)\) such calls.  The Hamming-weight workspace and the
temporary \(m\)-qubit register are absorbed by \(O(n/\log n)\).  Therefore the
standard all-to-all implementation has depth \(O(\log n)\) and uses
\(O(n/\log n)\) clean ancillary qubits.

In the dynamic all-to-all model, use
Corollary~\ref{cor:alltoall-dynamic-near-linear} for the Hamming-weight
computation and Lemma~\ref{lem:boolean-dynamic-alltoall} with \(k=m\) for the
Boolean post-processing.  The Hamming-weight computation has depth
\(O_\varepsilon(1)\) and uses
$
    O\!\left(n^{1+\varepsilon}\,\mathrm{polylog}\,n\right)
$
clean ancillary qubits.  The Boolean post-processing has constant depth and
uses
$
    O\!\left(2^m\,\mathrm{polylog}\,2^m\right)
    =
    O\!\left(n\,\mathrm{polylog}\,n\right)
$
clean ancillary qubits, which is absorbed by the near-linear Hamming-weight
bound for every fixed \(\varepsilon>0\).  Hence the dynamic all-to-all bounds
follow.

In the standard 2D model, use
Theorem~\ref{thm:hw-2d-polylog-ancilla} for the Hamming-weight computation and
Lemma~\ref{lem:boolean-standard-2d-borrowed} for the Boolean post-processing,
again with the constant-bit splitting described above.  The Hamming-weight
computation has depth \(O(\sqrt n)\) and uses \(O(\log^2 n)\) clean ancillary
qubits.  Each Boolean-evaluation call has input length \(m-c+1\), and hence
has depth
$
    O\!\left(\sqrt{2^{m-c+1}}\right)
    =
    O(\sqrt n),
$
using borrowed qubits supplied by the original input register and no
asymptotic clean ancillary workspace.  Any nearest-neighbor routing needed to
place the Hamming-weight register, the target qubit, and the borrowed input
qubits in the Boolean-evaluation patch is performed in depth \(O(\sqrt n)\)
on an \(O(n)\)-area grid and is reversed afterwards.  Since there are only
\(O(1)\) Boolean-evaluation calls, the total depth remains \(O(\sqrt n)\), and
the clean ancillary cost remains \(O(\log^2 n)\).

Finally, in the dynamic 2D model, use
Corollary~\ref{cor:2d-dynamic-near-linear} for the Hamming-weight computation
and Lemma~\ref{lem:boolean-dynamic-2d} with \(k=m\) for the Boolean
post-processing.  The Hamming-weight computation has depth \(O_\varepsilon(1)\)
and uses
$
    O\!\left(n^{1+\varepsilon}\,\mathrm{polylog}\,n\right)
$
clean ancillary qubits.  The Boolean post-processing has constant depth and
uses
$
    O\!\left(2^m\,\mathrm{polylog}\,2^m\right)
    =
    O\!\left(n\,\mathrm{polylog}\,n\right)
$
clean ancillary qubits, which is again absorbed by the near-linear
Hamming-weight bound for every fixed \(\varepsilon>0\).  Hence the dynamic
2D bounds follow.
\end{proof}

\section{Discussion and Open Problems}
\label{sec:discussion}

In this work, we presented quantum circuit constructions for two fundamental computational tasks, Hamming weight computation and symmetric Boolean functions. Our circuits are efficient in both circuit depth and ancilla count, and apply to both the standard and dynamic circuit models under all-to-all and 2D qubit connectivity. In most circuit models, the depth of our circuits is asymptotically optimal, and we also establish a depth--ancilla tradeoff that interpolates between efficient-depth and low-ancilla constructions. 

Several questions remain open. First, although our circuits achieve optimal depth in most circuit models, whether the depth-ancilla tradeoff itself is optimal is still unknown. 
Second, our analysis adopts the standard depth metric in which all Clifford gates and arbitrary single-qubit rotations are treated as unit cost. In the fault-tolerant setting, the dominant cost comes from non-Clifford resources, and it is therefore natural to ask what the minimum $T$-count and $T$-depth are for these two problems, and whether mid-circuit measurements can reduce $T$-depth as effectively as they reduce ordinary depth. Other directions, such as extensions to more general connectivity graphs, and to broader function classes including multi-output symmetric Boolean functions, are also worth exploring in future work.

\section*{Acknowledgment}

Wei Zi was supported by the Guangdong Provincial Quantum Science Strategic Initiative under Grant No.~GDZX2503001.

\bibliographystyle{plainnat}
\bibliography{ref}

\end{document}